\documentclass[11pt]{article}

\usepackage[margin=1in]{geometry}
\usepackage[T1]{fontenc}
\usepackage[utf8]{inputenc}
\usepackage{lmodern}
\usepackage{microtype}
\usepackage{parskip}
\usepackage{amsmath,amssymb,amsthm,mathtools,bm}
\usepackage{booktabs,threeparttable,array}
\usepackage{float}
\usepackage{enumitem}
\usepackage[authoryear,round]{natbib}
\usepackage[hidelinks]{hyperref}
\usepackage[nameinlink,noabbrev]{cleveref}

\newcommand{\R}{\mathbb{R}}
\newcommand{\E}{\mathbb{E}}
\newcommand{\Var}{\operatorname{var}}

\newcommand{\diag}{\operatorname{diag}}
\newcommand{\ones}{\mathbf{1}}

\newcommand{\rank}{\operatorname{rank}}

\newcommand{\vecop}{\operatorname{vec}}
\newcommand{\stackop}{\operatorname*{stack}}

\theoremstyle{plain}
\newtheorem{theorem}{Theorem}[section]
\newtheorem{proposition}[theorem]{Proposition}
\newtheorem{lemma}[theorem]{Lemma}
\newtheorem{corollary}[theorem]{Corollary}

\theoremstyle{definition}
\newtheorem{assumption}[theorem]{Assumption}

\theoremstyle{remark}
\newtheorem{remark}[theorem]{Remark}

\crefname{theorem}{Theorem}{Theorems}
\crefname{proposition}{Proposition}{Propositions}
\crefname{lemma}{Lemma}{Lemmas}
\crefname{corollary}{Corollary}{Corollaries}
\crefname{equation}{equation}{equations}

\title{Moment-Based Inference for Regression with Latent Dirichlet Covariates}
\author{Ziyu Jiang\thanks{Department of Economics, UCL. Email: \href{mailto:uctpzji@ucl.ac.uk}{\texttt{uctpzji@ucl.ac.uk}}.}}
\date{\today}

\begin{document}
\maketitle

\begin{abstract}
Topic models are often used as first-stage dimension-reduction tools before
regression, with estimated document-level topic shares treated as observed
covariates. This plug-in workflow creates two inferential difficulties:
valid inference requires a regular first-stage-to-second-stage expansion
that propagates topic-estimation uncertainty, and, at fixed document length,
a document's topic mixture is not consistently recoverable from its own
words even when the population topic matrix is known. Corrected spectral
moment methods for LDA provide a natural starting point: when the total
Dirichlet concentration parameter is known, low-order word moments can be
corrected to yield operators diagonal in the latent topic basis. We extend
this idea to downstream regression. Under a finite latent Dirichlet
allocation model with response residuals orthogonal to the low-order token
moments used for identification, response-weighted word moments admit the
same correction, and the resulting supervised operator identifies the
regression coefficient \(\beta\) directly, without estimating document-level
topic shares. The main theoretical obstacle is that the spectral correction
depends on the unknown total concentration \(\alpha_0\). We show that, for
\(k\ge3\) topics and under a generic finite-probe condition, \(\alpha_0\) is
identifiable by commutativity: at the true value, a family of corrected
word-moment operators commute, whereas away from the truth they generically
do not. This yields a feasible estimator and allows uncertainty in
\(\hat\alpha_0\) to be propagated into inference for \(\beta\). The estimator
is asymptotically linear as the number of documents grows with fixed document
length, with sandwich standard errors based on document-level moment
contributions. Simulations show near-nominal coverage where plug-in
topic-share regressions can undercover, and an application to top economics
journals illustrates contrast inference for latent topic effects.
\end{abstract}

\section{Introduction and Literature Review}\label{sec:intro}

Topic models are widely used as preprocessing tools for empirical
regression analysis. A corpus of high-dimensional text is first mapped
into a low-dimensional representation, typically estimated document-level
topic shares, and these topic-derived variables are then used as
regressors, controls, or predictors in a downstream model. Examples
include newspaper topics in regressions for political violence
\citep{MuellerRauh2018}, news-topic predictors in macroeconomic and
financial forecasting regressions \citep{LarsenThorsrud2019,Thorsrud2020},
topic-based measures of business news in vector autoregressions and
stock-market timing models \citep{BybeeKellyManelaXiu2024}, and topic
measures constructed from central-bank transcripts in econometric studies
of transparency and deliberation \citep{HansenMcMahonPrat2018}. In this
role, latent Dirichlet allocation and related topic models are best
viewed as dimension-reduction steps: they convert sparse word-count data
into a small number of latent coordinates intended for subsequent
statistical use.

This perspective connects topic-model preprocessing to the use of
estimated factors in factor-augmented regressions, where a high-dimensional
panel is compressed into a lower-dimensional latent representation before
a second-stage analysis \citep{StockWatson2002,Bai2003,BaiNg2006}. The
analogy is useful because it highlights the inferential issue. Once an
estimated latent representation is used as a regressor, uncertainty from
the first step must be reconciled with inference in the second step. In
the empirical topic-modeling workflow, this reconciliation is often
absent: researchers estimate topics, construct document-level topic
shares, and then condition on those estimated shares as if they were
ordinary observed covariates.

We consider a canonical finite-LDA version of this problem. Let document
\(i\) contain \(N\) word tokens \(x_{i1},\ldots,x_{iN}\), let \(k\) denote
the number of topics, let \(h_i\) denote its latent topic mixture, and let
\(O\) denote the topic matrix. Conditional on \(h_i\), the word tokens are
generated from the mixture \(Oh_i\). Suppose a scalar response can be written
as
\[
        Y_i=\beta^\top h_i+\varepsilon_i,
\]
where the residual is orthogonal to the low-order token moments used for
identification; a sufficient model-level condition is
\[
        \E(\varepsilon_i\mid h_i,x_{i1},\ldots,x_{iN})=0 .
\]
The object of inference is the downstream coefficient \(\beta\). We treat
the topic matrix, Dirichlet parameters and regression coefficient as fixed
unknown population parameters; the Dirichlet law for \(h_i\) is a
random-effects distribution for document heterogeneity, not a prior for a
posterior analysis over these unknown parameters. The usual plug-in workflow
estimates topic shares \(\hat h_i\) and then regresses \(Y_i\) on
\(\hat h_i\). This creates two difficulties.

The first is a first-stage regularity problem. Let \(n\) denote the number
of independent documents. For classical downstream inference, the relevant
object is not only the first-stage topic estimate but the entire map from
the observed corpus to the final regression coefficient. A standard
delta-method argument requires this map to be regular in the sense that
\[
    \sqrt n(\hat\beta-\beta)
    =
    n^{-1/2}\sum_{i=1}^n \phi_i+o_p(1),
\]
for an influence function \(\phi_i\) with consistently estimable variance.
Many topic-modeling procedures used in applied work, including Bayesian
and approximate-Bayesian implementations of LDA, produce posterior or
approximate-posterior summaries under chosen prior, hyperparameter and
computational specifications. Such summaries may be useful for exploration,
prediction or posterior inference under the specified model, but they do
not by themselves provide frequentist confidence intervals for the
downstream coefficient \(\beta\). Frequentist validity of the resulting
two-stage map must be established separately. In standard plug-in workflows,
the second-stage regression instead conditions on estimated topic shares as
if they were observed covariates. The resulting standard errors then use
only the second-stage regression variation and omit the first-stage
component of the influence function; unless that omitted component is
asymptotically negligible, the confidence intervals are not valid for the
sampling distribution of the full two-step procedure.

The second difficulty is more fundamental. Even if the population topic
matrix \(O\) were known, a document-specific topic mixture \(h_i\) cannot
be consistently recovered from the words in document \(i\) when document
length is fixed. The number of documents may grow, but the information
about any particular document's latent mixture remains bounded. Hence the
plug-in regressor \(\hat h_i\) contains non-vanishing document-level
error. This is the generated-regressor problem emphasized by
\citet{BattagliaChristensenHansenSacher2024} for regressions using
variables extracted from unstructured data. In the present setting it
arises even under perfect recovery of the topic matrix: increasing the
number of documents improves estimation of population moments, but it
does not make each individual \(h_i\) observed. A regression on \(\hat h_i\) can therefore have a wrong probability
limit or a persistent finite-document bias, not merely standard errors that
condition on a noisy generated regressor.

Integrated latent-variable likelihood procedures or fully Bayesian joint
models could also target \(\beta\) under a more complete specification by
integrating over the latent document mixtures. The contribution here is
complementary: we give a low-order moment route to frequentist confidence
intervals for \(\beta\), based on the supervised moment restrictions stated
below, without specifying a full likelihood for the joint distribution of
responses and words or a prior over the unknown population parameters. The
comparison with plug-in regressions is therefore aimed at the common two-step
topic-share regression workflow rather than at exhausting all possible
latent-variable estimators.

The spectral approach to LDA provides a useful starting point for
resolving both difficulties. \citet{AnandkumarGeHsuKakadeLiu2012} show
that, when the total Dirichlet concentration parameter is known, the
topic matrix can be recovered from low-order cross-token word moments
after applying Dirichlet-specific corrections. These corrections subtract
lower-order terms induced by the common document-level topic mixture,
leaving a population object that is diagonal in latent topic coordinates.
\citet{WangZhu2014} extend this corrected-moment logic to supervised
LDA under a related latent-mixture regression structure, introducing
response-weighted corrected moments that recover the regression weights.
\citet{RenWangZhu2018} develop this supervised spectral approach further,
giving both a two-stage method that recovers LDA parameters before
recovering the regression model and a single-phase method that jointly
recovers topic and regression parameters.

We take this supervised spectral moment structure as the starting point. Our
supervised identity uses the same response-weighted moment logic, but we
formulate it as a conceptually close and operationally different observed-space
operator. The response-weighted correction is formed in observed word space,
and after projection onto the recovered topic basis the representation
\(O^+H^y_{\alpha_0}O\) is diagonal with entries proportional to the components
of \(\beta\). The distinction is therefore not the underlying supervised-LDA
moment model, but the inferential target and estimation problem: we study
frequentist inference for a downstream regression coefficient at fixed document
length, without estimating document-level topic shares.

The remaining roadblock, common to the spectral LDA and supervised spectral
LDA procedures just discussed, is that the corrected moments are indexed by
the total Dirichlet concentration parameter. Write this concentration as
\(\alpha_0\). In these procedures, \(\alpha_0\) is supplied as an input when
forming the corrected moments. In empirical work this quantity is rarely known,
and treating it as a tuning parameter is unsatisfactory for two reasons. First,
the downstream procedure then rests on a choice of \(\alpha_0\) rather than
on an identification argument for \(\alpha_0\) from the observable word
distribution. Practical selection rules are therefore difficult to
interpret as estimators of a population parameter. Second, even if a
particular rule performs well computationally, its uncertainty has no
automatic route into the final confidence interval for \(\beta\). We show
that, when the number of topics satisfies \(k\ge3\) and a generic finite-probe
condition holds, \(\alpha_0\) is itself identified from observable moments:
at the true concentration value, a family of corrected word-moment operators
commute, whereas away from the truth they generically fail to commute. This
gives a feasible estimator of \(\alpha_0\) and allows its uncertainty to be
propagated through the estimator of the downstream coefficient.

We work in a fixed-document-length asymptotic regime: the number of
independent documents grows, while the number of words in each document
may remain bounded. In this regime, the target is inference on the
population coefficient \(\beta\), not recovery of every document-specific
topic mixture. Within this setting, the paper makes three contributions.

First, for known total concentration, we give an observed-space operator
formulation of the supervised spectral moment identity. This identity starts
from the same response-weighted corrected moments as the supervised spectral
LDA literature, but casts them in a form tailored to downstream inference:
the supervised operator is an observed-space matrix whose topic-basis
representation \(O^+H^y_{\alpha_0}O\) is diagonal, and those diagonal entries
are proportional to the components of \(\beta\). This identifies the downstream
coefficient directly from observable word--response moments, without
constructing document-level topic-share regressors.

Second, we identify and estimate the total Dirichlet concentration parameter.
The spectral LDA and supervised spectral LDA procedures discussed above form
their corrected moments conditional on \(\alpha_0\). We show that, when the
number of topics is at least three and a generic probe condition holds,
\(\alpha_0\) is the unique value for which a finite family of corrected
word-moment operators commute. This gives a feasible estimator of \(\alpha_0\)
from the same low-order word moments used for topic recovery.

Third, relative to the supervised spectral recovery results in
\citet{WangZhu2014} and \citet{RenWangZhu2018}, we establish frequentist
inference for the downstream coefficient in a fixed-dimensional regime. The
primitive empirical quantities are distinct-token word moments and
response-weighted word moments, averaged across independent documents. With
fixed vocabulary dimension \(d\), fixed topic dimension \(k\), and fixed document
length \(N\), or after conditioning on a fixed admissible compression dimension,
these moments satisfy an ordinary root-\(n\) central limit theorem over
documents. The concentration-parameter, spectral, and supervised coefficient
maps are smooth under rank and separation conditions. The resulting estimator
of \(\beta\) is asymptotically linear with a feasible sandwich variance
estimator, so the contribution is a first-order sampling expansion and feasible
standard errors rather than only consistency or sample-complexity recovery.

The Monte Carlo experiments are designed to match this inferential logic.
They evaluate the finite-sample performance of the proposed confidence
intervals and compare them with plug-in regressions using estimated
document-topic shares. The proposed intervals have coverage close to
nominal in the designs considered. Plug-in regressions undercover
substantially, including a known-\(O\) plug-in reconstruction procedure in
which the topic matrix is supplied but document-level mixtures are still
reconstructed from finite document counts. This isolates the fixed-document
generated-regressor problem: increasing the number of documents improves
population moment estimation, but it does not make the document-level topic
mixtures observed. The simulations are intended to isolate this mechanism
under favorable rank, separation and probe conditions, not to exhaust all
finite-sample stress tests of the feasible workflow.

The rest of the paper is organized as follows. Section~\ref{sec:model}
introduces the finite-LDA model, the downstream regression target, and
the observable word and word--response moments. Section~\ref{sec:identification}
derives the corrected moment operators, proves population identification
of the topic matrix and the downstream coefficient, and establishes
identification of \(\alpha_0\) through commutativity. Section~\ref{sec:estimation}
defines the proposed moment-based estimators and gives the asymptotic distribution and
sandwich variance estimator for the downstream coefficient. Section~\ref{sec:simulation}
reports Monte Carlo evidence on finite-sample performance and on the
failure of plug-in topic regressions at fixed document length. Section~\ref{sec:application}
gives a real-data illustration using articles from the five general-interest
economics journals. Section~\ref{sec:discussion} discusses limitations and
extensions. Proofs and extensions, including fixed-dimensional observed
controls and split-sample linear compression for high-dimensional preprocessing,
are collected in the Supplement.

\section{Model and target parameter}\label{sec:model}

This section specifies the finite-LDA sampling experiment and the downstream regression target. The point of the formulation is to separate the latent document-level topic mixture, which is useful for defining the model, from the population regression coefficient, which is the object of inference. The corrected spectral moments used for identification are introduced in Section~\ref{sec:identification}.

\subsection{Observed data and latent topic structure}

We observe independent documents
\[
    (Y_i,x_{i1},\ldots,x_{iN}),\qquad i=1,\ldots,n,
\]
where \(Y_i\in\R\) is a scalar document-level response. The vocabulary has
\(d\) terms, and each word token is represented by a one-hot vector
\(x_{ij}\in\{e_1,\ldots,e_d\}\subset\R^d\). The main theory is stated for
a common document length \(N\ge3\), fixed as \(n\to\infty\). The condition
\(N\ge3\) is used because the spectral construction relies on third-order
cross-token moments. Variable document lengths \(N_i\ge3\) can be handled by
normalizing the within-document moment averages document by document, as
recorded in Appendix~\ref{app:variable-lengths}; we maintain the same
sampled-document LDA moment structure throughout.

Let \(k\) denote the number of topics, and let
\[
    O=[O_{:1},\ldots,O_{:k}]\in\R^{d\times k}
\]
denote the topic matrix. Each column \(O_{:\ell}\) lies in the vocabulary
simplex \(\Delta^{d-1}\) and gives the word distribution for topic
\(\ell\). The matrix \(O\) and the Dirichlet parameter below are treated
as fixed unknown population parameters.

For document \(i\), the latent topic mixture is
\[
    h_i\sim\operatorname{Dirichlet}(\alpha),
    \qquad
    \alpha=(\alpha_1,\ldots,\alpha_k)^\top\in(0,\infty)^k,
\]
with total concentration
\begin{equation}\label{eq:alpha0-def}
    \alpha_0:=\sum_{\ell=1}^k\alpha_\ell .
\end{equation}
The usual LDA sampling scheme draws latent topic labels
\[
    z_{ij}\mid h_i \sim \operatorname{categorical}(h_i),
    \qquad j=1,\ldots,N,
\]
independently across token positions, and then draws words according to
\[
    x_{ij}\mid z_{ij}=\ell
    \sim \operatorname{categorical}(O_{:\ell}) .
\]
The topic labels \(z_{ij}\) are not observed and will not be used
directly. Integrating them out gives the equivalent conditional word
distribution
\begin{equation}\label{eq:cond-prob}
    \Pr(x_{ij}=e_v\mid h_i)=(Oh_i)_v,
    \qquad v=1,\ldots,d .
\end{equation}
Thus, conditional on \(h_i\), the word tokens are independent and
identically distributed with conditional mean
\begin{equation}\label{eq:cond-mean}
    \E[x_{ij}\mid h_i]=Oh_i .
\end{equation}
Marginally, tokens in the same document are dependent through the common
latent mixture \(h_i\), which is the dependence exploited by the
cross-token moment identities below.

The document-specific vector \(h_i\) is not observed. Because \(N\) is
fixed, the analysis does not require, and does not assume, that \(h_i\)
can be consistently recovered for each document. The inferential target
below is instead a finite-dimensional population parameter identified
from the joint law of \((Y_i,x_{i1},\ldots,x_{iN})\).

\subsection{Downstream response model}

The response is linked to the latent topic mixture by
\begin{equation}\label{eq:supervised-linear-model}
    Y_i=\beta^\top h_i+\varepsilon_i,
\end{equation}
for some \(\beta\in\R^k\), with \(\E[Y_i^2]<\infty\). The primitive
restriction used by the supervised moment identities is low-order
response-token orthogonality over the token positions entering the averaged moments:
\begin{equation}\label{eq:response-token-orthogonality}
    \E(\varepsilon_i)=0,\qquad
    \E(x_{ia}\varepsilon_i)=0\quad\text{for all }a,\qquad
    \E(x_{ia}x_{ib}^\top\varepsilon_i)=0\quad\text{for all }a\ne b .
\end{equation}
A sufficient model-level condition is
\begin{equation}\label{eq:strong-response-condition}
    \E(\varepsilon_i\mid h_i,x_{i1},\ldots,x_{iN})=0,
\end{equation}
or equivalently
\(\E(Y_i\mid h_i,x_{i1},\ldots,x_{iN})=\beta^\top h_i\). The weaker
condition \(\E(Y_i\mid h_i)=\beta^\top h_i\) alone is not enough for the
response-weighted word moments used below if residual variation is
systematically related to realized token choices beyond \(h_i\).

It should be noted that, since \(\ones^\top h_i=1\), a separate intercept is not identified from a
common shift of all topic coefficients. If a model is written as
\(c+\tilde\beta^\top h_i\), then it is observationally equivalent to
\((\tilde\beta+c\ones)^\top h_i\). We therefore adopt the normalization that absorbs the intercept into
\(\beta\). This is the same normalization issue that arises in a regression
with a full set of category indicators: the common level is a convention, while
contrasts such as \(\beta_a-\beta_b\) are invariant and often the substantively
meaningful objects. Topic labels are arbitrary, so \(\beta\) is always
interpreted in the same ordering as the columns of \(O\). Appendix~\ref{app:observed-controls}
gives the corresponding fixed-dimensional extension with observed controls.

\subsection{Observable low-order moments}

All population moments below refer to a generic document, and we suppress the document index. By exchangeability, any distinct token positions have the same joint distribution, so we use positions \(1,2,3\) for notation. Write
\[
    \pi:=\frac{\alpha}{\alpha_0},
    \qquad
    D:=\diag(\alpha).
\]
The first two Dirichlet moments are
\begin{equation}\label{eq:dir-first}
    \E[h]=\pi,
\end{equation}
and
\begin{equation}\label{eq:dir-second-matrix}
    \E[hh^\top]
    =\frac{D+\alpha\alpha^\top}{\alpha_0(\alpha_0+1)} .
\end{equation}
The corresponding first and second cross-token word moments are
\begin{equation}\label{eq:first-moment}
    \mu:=\E[x_1]=O\pi,
\end{equation}
\begin{equation}\label{eq:second-moment}
\begin{aligned}
    M_2:=\E[x_1x_2^\top]
    &=O\E[hh^\top]O^\top  \\
    &=\frac{1}{\alpha_0(\alpha_0+1)}ODO^\top
      +\frac{\alpha_0}{\alpha_0+1}\mu\mu^\top .
\end{aligned}
\end{equation}
The use of distinct positions in \(M_2\) is important: \(M_2\) is a cross-token moment, not the second moment of a single multinomial draw.

For a contraction direction \(\eta\in\R^d\), define the flattened third cross-token moment
\begin{equation}\label{eq:T-eta-def}
    T(\eta):=\E\{x_1x_2^\top\langle x_3,\eta\rangle\}.
\end{equation}
Let \(\gamma:=O^\top\eta\). Conditional independence and \eqref{eq:cond-mean} give
\begin{equation}\label{eq:T-eta-pullout}
    T(\eta)=O\,\E\{hh^\top(\gamma^\top h)\}\,O^\top .
\end{equation}
The contracted third Dirichlet moment is
\begin{equation}\label{eq:dir-contracted-third}
\begin{aligned}
\E\{hh^\top(\gamma^\top h)\}
=\frac{1}{\alpha_0(\alpha_0+1)(\alpha_0+2)}
\Big
(& (\alpha^\top\gamma)\alpha\alpha^\top
  +\alpha\alpha^\top\diag(\gamma)
  +\diag(\gamma)\alpha\alpha^\top  \\
& + (\alpha^\top\gamma)D
  +2\diag(\alpha\circ\gamma)
\Big).
\end{aligned}
\end{equation}
Thus
\begin{equation}\label{eq:T-eta-explicit}
\begin{aligned}
T(\eta)
=\frac{1}{\alpha_0(\alpha_0+1)(\alpha_0+2)}
O\Big
(& (\alpha^\top\gamma)\alpha\alpha^\top
  +\alpha\alpha^\top\diag(\gamma)
  +\diag(\gamma)\alpha\alpha^\top  \\
& + (\alpha^\top\gamma)D
  +2\diag(\alpha\circ\gamma)
\Big)O^\top .
\end{aligned}
\end{equation}
The corrected moments in the next section subtract lower-order terms induced by the common document mixture; at \(\tau=\alpha_0\), this removes the rank-one term in \eqref{eq:second-moment} and all lower-order third-moment terms in \eqref{eq:T-eta-explicit}, leaving only the diagonal factor proportional to \(\diag(\alpha\circ\gamma)\).

Finally, define the supervised observable moments
\begin{equation}\label{eq:supervised-response-moments}
    m_y:=\E[Y],
    \qquad
    v_y:=\E[x_1Y],
    \qquad
    T^y:=\E[x_1x_2^\top Y].
\end{equation}
Under \eqref{eq:supervised-linear-model} and
\eqref{eq:response-token-orthogonality},
\begin{equation}\label{eq:supervised-moment-factorizations}
    m_y=\beta^\top\pi,
    \qquad
    v_y=O\E[hh^\top]\beta,
    \qquad
    T^y=O\E\{hh^\top(\beta^\top h)\}O^\top .
\end{equation}
Consequently, whenever \(O\) has full column rank, there exists \(\eta_\beta\in\R^d\) such that
\begin{equation}\label{eq:eta-beta-bridge}
    O^\top\eta_\beta=\beta,
\end{equation}
and the supervised moments satisfy the bridge identities
\begin{equation}\label{eq:supervised-bridge-identities}
    m_y=\langle \eta_\beta,\mu\rangle,
    \qquad
    v_y=M_2\eta_\beta,
    \qquad
    T^y=T(\eta_\beta).
\end{equation}
These identities are the population reason why the downstream coefficient can be recovered from observable word--response moments without first constructing document-level estimates of \(h_i\).

\section{Population identification through corrected operators}\label{sec:identification}

This section gives the population identities on which the estimators are based. The central point is that the document-level mixture \(h_i\) need not be recovered document by document. Instead, low-order observable moments can be corrected so that the resulting operators are diagonal in the latent topic coordinates. One operator recovers the topic directions; its supervised analogue recovers the downstream coefficient vector. A second consequence of the same diagonalization is an identifying restriction for the unknown concentration mass \(\alpha_0\).

Throughout the section, expectations are under the model of Section~\ref{sec:model}. We write \(\tau>0\) for a candidate value of the total concentration and reserve \(\alpha_0\) for the true value. Let
\[
    \mu=\E[x_1],\qquad
    M_2=\E[x_1x_2^\top],\qquad
    T(\eta)=\E\{x_1x_2^\top\langle x_3,\eta\rangle\},
\]
where \(\eta\in\R^d\). Define the \(\tau\)-corrected second moment
\begin{equation}\label{eq:id-Btau}
    B_\tau
    :=M_2-\frac{\tau}{\tau+1}\mu\mu^\top,
\end{equation}
and the \(\tau\)-corrected contracted third moment
\begin{equation}\label{eq:id-Atau}
\begin{aligned}
    A_\tau(\eta)
    := {}&T(\eta)
    -\frac{\tau}{\tau+2}
    \{M_2\eta\mu^\top+\mu\eta^\top M_2+\langle \eta,\mu\rangle M_2\}  \\
    &\quad
    +\frac{2\tau^2}{(\tau+1)(\tau+2)}
    \langle \eta,\mu\rangle\mu\mu^\top .
\end{aligned}
\end{equation}
The associated observed-space operator is
\begin{equation}\label{eq:id-Htau}
    H_\tau(\eta):=A_\tau(\eta)B_\tau^+,
\end{equation}
where \(B_\tau^+\) is the Moore--Penrose inverse. The operator is generally not symmetric, so eigenvectors below are right eigenvectors.

\subsection{Corrected LDA moments}\label{subsec:corrected-moments}

The following lemma is the algebraic core of the paper. It states that the correction in \eqref{eq:id-Btau}--\eqref{eq:id-Atau} removes the off-diagonal Dirichlet terms exactly at \(\tau=\alpha_0\). Let
\[
    D:=\diag(\alpha),\qquad
    C_2:=\alpha_0(\alpha_0+1),\qquad
    C_3:=\alpha_0(\alpha_0+1)(\alpha_0+2).
\]

\begin{lemma}[Corrected moment factorization]\label{lem:id-corrected-factorization}
Assume \(O\in\R^{d\times k}\) has full column rank and \(\alpha\in(0,\infty)^k\). For any \(\tau>0\) and \(\eta\in\R^d\), put \(w=O^\top\eta\). Then
\begin{equation}\label{eq:id-B-factor}
    B_\tau=O S_B(\tau)O^\top,
    \qquad
    S_B(\tau)=\frac{1}{C_2}D+
    \frac{\alpha_0-\tau}{\alpha_0^2(\alpha_0+1)(\tau+1)}\alpha\alpha^\top .
\end{equation}
Moreover, \(S_B(\tau)\) is positive definite for every \(\tau>0\), with
\begin{equation}\label{eq:id-SB-inv}
    S_B(\tau)^{-1}=C_2D^{-1}-(\alpha_0-\tau)\ones\ones^\top .
\end{equation}
The third-order correction satisfies
\begin{equation}\label{eq:id-A-factor-general}
    A_\tau(\eta)=O S_A(\tau;w)O^\top,
\end{equation}
where \(S_A(\tau;w)\) is linear in \(w\). At the true concentration,
\begin{equation}\label{eq:id-true-diagonal-moments}
    B_{\alpha_0}=O\frac{D}{C_2}O^\top,
    \qquad
    A_{\alpha_0}(\eta)=O\frac{2}{C_3}\diag(\alpha\circ w)O^\top .
\end{equation}
Consequently,
\begin{equation}\label{eq:id-H-true-diagonal}
    H_{\alpha_0}(\eta)
    =O\left\{\frac{2}{\alpha_0+2}\diag(O^\top\eta)\right\}O^+ .
\end{equation}
\end{lemma}

The proof is a direct calculation from the Dirichlet second and third moments and is given in the Supplement. The number of topics is also identified at this point: since \(\E[hh^\top]\) is positive definite and \(O\) has full column rank, \(\rank(M_2)=k\).

\subsection{Topic directions and supervised coefficients when \texorpdfstring{\(\alpha_0\)}{alpha0} is known}\label{subsec:known-alpha-identification}

Equation~\eqref{eq:id-H-true-diagonal} gives the usual spectral recovery of the topic directions, but in observed coordinates rather than whitened coordinates. If \(\eta\) is chosen so that the coordinates of \(O^\top\eta\) are distinct and nonzero, the right eigenvectors associated with the \(k\) nonzero simple eigenvalues of \(H_{\alpha_0}(\eta)\) are exactly the topic vectors.

\begin{theorem}[Topic identification with known concentration]\label{thm:id-known-alpha-topic}
Assume \(O\in\R^{d\times k}\) has full column rank and \(\alpha\in(0,\infty)^k\). Let \(\eta\) be drawn from any distribution on \(\R^d\) that is absolutely continuous with respect to Lebesgue measure. Then, with probability one over \(\eta\), the matrix \(H_{\alpha_0}(\eta)\) has exactly \(k\) nonzero simple eigenvalues,
\begin{equation}\label{eq:id-topic-eigenvalues}
    \lambda_j(\eta)=\frac{2}{\alpha_0+2}(O^\top\eta)_j,
    \qquad j=1,\ldots,k,
\end{equation}
and the associated right eigenvectors are \(O_{:1},\ldots,O_{:k}\). Thus \(O\) is identified up to column permutation and scaling; the simplex normalization \(\ones^\top O_{:j}=1\) fixes the scaling.
\end{theorem}

The downstream coefficient is identified by a parallel supervised correction. Using the response moments from Section~\ref{sec:model}, define for a candidate \(\tau\)
\begin{equation}\label{eq:id-Aytau}
\begin{aligned}
    A_\tau^y
    :={}&T^y
    -\frac{\tau}{\tau+2}
    \{v_y\mu^\top+\mu v_y^\top+m_y M_2\}  \\
    &\quad
    +\frac{2\tau^2}{(\tau+1)(\tau+2)}m_y\mu\mu^\top,
\end{aligned}
\end{equation}
and
\begin{equation}\label{eq:id-Hytau}
    H_\tau^y:=A_\tau^y B_\tau^+ .
\end{equation}
The key identity is obtained at \(\tau=\alpha_0\).

\begin{theorem}[Direct identification of the downstream coefficient]\label{thm:id-beta-supervised}
Assume the model of Section~\ref{sec:model}, including the response-token
orthogonality conditions in \eqref{eq:response-token-orthogonality}. If
\(O\) has full column rank, then
\begin{equation}\label{eq:id-Ay-diagonal}
    A_{\alpha_0}^y
    =O\frac{2}{C_3}\diag(\alpha\circ\beta)O^\top
\end{equation}
and
\begin{equation}\label{eq:id-Hy-diagonal}
    H_{\alpha_0}^y
    =O\left\{\frac{2}{\alpha_0+2}\diag(\beta)\right\}O^+ .
\end{equation}
Consequently, once the topic columns have been labeled by the unsupervised spectral step,
\begin{equation}\label{eq:id-beta-projection}
    \beta_j
    =\frac{\alpha_0+2}{2}\{O^+H_{\alpha_0}^yO\}_{jj},
    \qquad j=1,\ldots,k .
\end{equation}
This identification does not require the entries of \(\beta\) to be distinct and does not require document-level estimates of \(h_i\).
\end{theorem}

The theorem follows from the bridge identities
\(m_y=\langle\eta_\beta,\mu\rangle\), \(v_y=M_2\eta_\beta\), and
\(T^y=T(\eta_\beta)\) for any \(\eta_\beta\) satisfying
\(O^\top\eta_\beta=\beta\). These identities are implied by the low-order
orthogonality conditions in \eqref{eq:response-token-orthogonality}; the
stronger conditional mean restriction \eqref{eq:strong-response-condition}
is only one sufficient primitive condition. Thus
\(A_{\alpha_0}^y=A_{\alpha_0}(\eta_\beta)\), and
Lemma~\ref{lem:id-corrected-factorization} applies. Equation~\eqref{eq:id-Hy-diagonal}
shows that the supervised operator is diagonal in the topic basis. If the
entries of \(\beta\) are distinct, this relation may be read as an eigenvalue
statement. For inference, however, we use the projection formula
\eqref{eq:id-beta-projection}, which avoids the complications caused by
repeated or zero coefficient values once the topic basis has been fixed.

\subsection{Identification of the concentration mass}\label{subsec:alpha0-identification}

The preceding identities still require the total concentration \(\alpha_0\)
as an input. This is a central practical limitation of corrected spectral LDA
moments: the correction that diagonalizes the word moments depends on a
quantity that is rarely known in applications. We now show that the same
operator representation identifies \(\alpha_0\) from observable word moments.
The identifying restriction is commutativity.

At the true value, \eqref{eq:id-H-true-diagonal} implies that every corrected
operator is diagonal in the same topic basis. Hence any two such operators
commute: for every pair \(\eta_1,\eta_2\),
\begin{equation}\label{eq:id-commute-true}
    [H_{\alpha_0}(\eta_1),H_{\alpha_0}(\eta_2)]=0,
\end{equation}
where \([R,S]=RS-SR\). We show that this property is generically unique to
the true concentration: away from \(\alpha_0\), mean-orthogonal contractions
produce nonzero commutators. Thus commutativity supplies an identifying
restriction for the concentration mass. Let
\begin{equation}\label{eq:id-projection-def}
    P_\mu:=I_d-\frac{\mu\mu^\top}{\|\mu\|^2}.
\end{equation}
If \(\eta\in\mu^\perp\), then \(w=O^\top\eta\in\alpha^\perp\), because
\[
    \alpha^\top O^\top\eta=\alpha_0\mu^\top\eta=0.
\]

\begin{theorem}[Identification of \texorpdfstring{\(\alpha_0\)}{alpha0} by commutativity]\label{thm:id-alpha-commutativity}
Assume \(k\ge3\), \(O\in\R^{d\times k}\) has full column rank, and \(\alpha\in(0,\infty)^k\). For \(\tau>0\), write
\[
    H_\tau(\eta)=O T_\tau(O^\top\eta)O^+,
    \qquad
    T_\tau(w):=S_A(\tau;w)S_B(\tau)^{-1}.
\]
Then the following statements hold.

\begin{enumerate}[label=\textnormal{(\roman*)},leftmargin=2.5em]
\item At \(\tau=\alpha_0\), the family \(\{H_{\alpha_0}(\eta):\eta\in\R^d\}\) is pairwise commuting.

\item If \(\tau\ne\alpha_0\) and \(w_1,w_2\in\alpha^\perp\) are non-collinear, then
\begin{equation}\label{eq:id-latent-commutator}
    [T_\tau(w_1),T_\tau(w_2)]
    =c(\tau,\alpha_0)
    \{(\alpha\circ w_1)w_2^\top-(\alpha\circ w_2)w_1^\top\},
\end{equation}
where
\begin{equation}\label{eq:id-c-commutator}
    c(\tau,\alpha_0)
    =\frac{4(\alpha_0-\tau)(\alpha_0\tau+\alpha_0+\tau)}
    {\alpha_0(\alpha_0+1)(\alpha_0+2)^2(\tau+2)^2} .
\end{equation}
In particular, \([T_\tau(w_1),T_\tau(w_2)]\ne0\), and hence any \(\eta_1,\eta_2\in\mu^\perp\) satisfying \(O^\top\eta_j=w_j\) give
\[
    [H_\tau(\eta_1),H_\tau(\eta_2)]\ne0 .
\]
\end{enumerate}
Consequently, \(\alpha_0\) is the unique value of \(\tau>0\) for which the full family \(\{H_\tau(\eta):\eta\in\mu^\perp\}\) is pairwise commuting.
\end{theorem}

The restriction \(k\ge3\) is used only here. It ensures that \(\alpha^\perp\) has dimension at least two, so non-collinear mean-orthogonal contractions exist. The present commutator argument is therefore formulated for topic models with at least three topics.

For estimation, one does not need to check commutativity over all \(\eta\in\mu^\perp\). Let \(s_1,\ldots,s_L\in\R^d\), define
\begin{equation}\label{eq:id-finite-probes}
    v_\ell=P_\mu s_\ell,
    \qquad \ell=1,\ldots,L,
\end{equation}
and let \(\mathcal I\subset\{(\ell,q):1\le \ell<q\le L\}\) be non-empty. Consider the population finite-probe criterion
\begin{equation}\label{eq:id-pop-criterion}
    Q_{\mathcal I}(\tau)
    :=\sum_{(\ell,q)\in\mathcal I}
    \big\|[H_\tau(v_\ell),H_\tau(v_q)]\big\|_F^2 .
\end{equation}
If there is at least one pair \((\ell,q)\in\mathcal I\) for which \(O^\top v_\ell\) and \(O^\top v_q\) are non-collinear, then
\begin{equation}\label{eq:id-Q-identifies-alpha}
    Q_{\mathcal I}(\tau)=0
    \quad\Longleftrightarrow\quad
    \tau=\alpha_0 .
\end{equation}
If two raw probes \(s_1,s_2\) are drawn independently from distributions that are absolutely continuous on \(\R^d\), this non-collinearity condition holds with probability one. The sample estimator in Section~\ref{sec:estimation} is the plug-in version of \eqref{eq:id-pop-criterion}, using empirical moments and empirical mean-orthogonal projections.

The commutativity result also has a purely unsupervised implication. Combined
with the corrected-moment eigenvector identity in Theorem~\ref{thm:id-known-alpha-topic},
it implies that the finite-dimensional LDA model is identified from low-order
word moments without supplying the total concentration as an input. We record
this consequence explicitly.

\begin{corollary}[Unsupervised finite-LDA identification]\label{cor:unsupervised-lda-identification}
Under the assumptions of \Cref{thm:id-alpha-commutativity}, the word-side parameters of the finite LDA model are identified from the joint law of three distinct tokens, without supplying \(\alpha_0\) as an input. In particular, \(k=\rank(M_2)\), \(\alpha_0\) is identified by the commutativity condition in \Cref{thm:id-alpha-commutativity}, the topic matrix \(O\) is identified up to column permutation by \Cref{thm:id-known-alpha-topic} using any contraction direction with distinct nonzero latent coordinates, and the simplex normalizations \(\ones^\top O_{:j}=1\) fix the column scales. The Dirichlet parameter is then identified by
\[
    \pi=O^+\mu,\qquad \alpha=\alpha_0\pi .
\]
\end{corollary}

\begin{proof}
The rank statement follows from \(M_2=O\E(hh^\top)O^\top\), since \(\E(hh^\top)\) is positive definite and \(O\) has full column rank. The commutator condition identifies \(\alpha_0\). Given \(\alpha_0\), \Cref{thm:id-known-alpha-topic} identifies the columns of \(O\) up to permutation and scale, and the simplex normalizations fix the scale. Finally, \(\mu=O\alpha/\alpha_0\), so \(\pi=\alpha/\alpha_0=O^+\mu\), and \(\alpha=\alpha_0\pi\). Equivalently, the same correction can be viewed as the supervised bridge with pseudo-response \(Y_\eta=\langle x_3,\eta\rangle\), using \(x_1,x_2\) for the word moments. Conditional independence of distinct tokens gives
\[
    \E(Y_\eta)=\langle\eta,\mu\rangle,\qquad
    \E(x_1Y_\eta)=M_2\eta,\qquad
    \E(x_1x_2^\top Y_\eta)=T(\eta),
\]
so the corrected pseudo-response operator is exactly \(H_{\alpha_0}(\eta)\).
\end{proof}

\section{Estimation and inference}\label{sec:estimation}

This section turns the population identities of Section~\ref{sec:identification} into a feasible estimator. The construction has three steps. First, estimate the observable cross-token moments by within-document averages over distinct token positions. Second, estimate the concentration mass \(\alpha_0\) by minimizing a finite set of empirical commutators. Third, recover the topic matrix from one unsupervised spectral operator and recover the downstream coefficient from the supervised operator. The limit theory treats these steps as one smooth map of empirical moments, after rank selection and topic order have been determined.

\subsection{Empirical cross-token moments}\label{subsec:empirical-moments}

For document \(i\), define the within-document empirical mean
\begin{equation}\label{eq:est-mu-i}
    \hat\mu_i:=\frac1N\sum_{a=1}^N x_{ia}.
\end{equation}
The second cross-token moment is estimated by averaging over ordered distinct pairs,
\begin{equation}\label{eq:est-M2-i}
    \hat M_{2,i}
    :=\frac{1}{N(N-1)}\sum_{a\ne b}x_{ia}x_{ib}^\top .
\end{equation}
Because ordered pairs include both \((a,b)\) and \((b,a)\), \(\hat M_{2,i}\)
is symmetric. Equivalently, if \(c_i=\sum_{a=1}^N x_{ia}\) is the word-count vector, then
\begin{equation}\label{eq:est-M2-counts}
    \hat M_{2,i}
    =\frac{c_ic_i^\top-\diag(c_i)}{N(N-1)} .
\end{equation}
For a contraction direction \(v\in\R^d\), the third cross-token moment is estimated by
\begin{equation}\label{eq:est-T-i}
    \hat T_i(v)
    :=\frac{1}{N(N-1)(N-2)}
      \sum_{\substack{a,b,c=1\\ a,b,c\;\mathrm{distinct}}}^N
      x_{ia}x_{ib}^\top\langle x_{ic},v\rangle .
\end{equation}
The corresponding sample averages are
\begin{equation}\label{eq:est-raw-averages}
    \hat\mu:=\frac1n\sum_{i=1}^n\hat\mu_i,
    \qquad
    \hat M_2:=\frac1n\sum_{i=1}^n\hat M_{2,i},
    \qquad
    \hat T(v):=\frac1n\sum_{i=1}^n\hat T_i(v).
\end{equation}
The hats in this section refer to the sample size \(n\), which is suppressed when no ambiguity can arise. The estimators in \eqref{eq:est-raw-averages} are unbiased for \(\mu\), \(M_2\), and \(T(v)\), respectively. The use of distinct token positions removes the single-token multinomial contribution that would otherwise appear on the diagonal of the second moment.

The supervised moments are estimated in the same way. Define
\begin{equation}\label{eq:est-supervised-moments}
    \hat m_y:=\frac1n\sum_{i=1}^nY_i,
    \qquad
    \hat v_y:=\frac1n\sum_{i=1}^nY_i\hat\mu_i,
    \qquad
    \hat T^y:=\frac1n\sum_{i=1}^nY_i\hat M_{2,i}.
\end{equation}
Thus \(\hat T^y\) is the response-weighted analogue of the second cross-token moment. The full third-order tensor need not be stored in computation; only the finite number of contractions \(\hat T(v)\) used by the estimator are required.

\emph{Large-vocabulary preprocessing.} The formulae above are written in the
original vocabulary dimension. In large vocabularies, one can instead compress
tokens linearly before forming the corrected moments. Appendix~\ref{app:pca-compression}
shows that if a fixed-dimensional compression matrix does not lose a topic
direction, in the sense that \(\rank(R^\top O)=k\), then the corrected moment
identities and the coefficient \(\beta\) are unchanged. When observed controls
are included after compression, the control adjustment also requires the latent
topic-share scale. Appendix~\ref{app:pca-compression} gives a compressed
scale-recovery step based on the compressed first and corrected second moments;
this step is used before constructing the control cross-moment \(M_{qh}\). The
split-sample PCA version estimates an \(m\)-dimensional compression, with fixed
\(m\ge k\), from the uncentered cross-token second moment on an independent
split, holds it fixed, and runs the moment estimator on the compressed tokens in
the second split. This is intended as a compatibility result for fixed \(k\) and
fixed \(m\), not a rigorous growing-vocabulary or growing-rank asymptotic theory.

\subsection{Plug-in corrected operators}\label{subsec:plugin-operators}

For each candidate \(\tau>0\), define the empirical corrected moments
\begin{equation}\label{eq:est-Btau}
    \hat B_\tau
    :=\hat M_2-\frac{\tau}{\tau+1}\hat\mu\hat\mu^\top,
\end{equation}
and
\begin{equation}\label{eq:est-Atau}
\begin{aligned}
    \hat A_\tau(v)
    :={}&\hat T(v)
    -\frac{\tau}{\tau+2}
    \{\hat M_2v\hat\mu^\top+\hat\mu v^\top\hat M_2
      +\langle v,\hat\mu\rangle\hat M_2\}  \\
    &\quad
    +\frac{2\tau^2}{(\tau+1)(\tau+2)}
    \langle v,\hat\mu\rangle\hat\mu\hat\mu^\top .
\end{aligned}
\end{equation}
Similarly, the supervised corrected moment is
\begin{equation}\label{eq:est-Aytau}
\begin{aligned}
    \hat A_\tau^y
    :={}&\hat T^y
    -\frac{\tau}{\tau+2}
    \{\hat v_y\hat\mu^\top+\hat\mu\hat v_y^\top+\hat m_y\hat M_2\}  \\
    &\quad
    +\frac{2\tau^2}{(\tau+1)(\tau+2)}
    \hat m_y\hat\mu\hat\mu^\top .
\end{aligned}
\end{equation}
These are the sample analogues of \eqref{eq:id-Btau}, \eqref{eq:id-Atau}, and \eqref{eq:id-Aytau}.

The corrected second moment has rank \(k\) in population, so the inverse in the empirical operator is taken on an estimated rank-\(k\) subspace. Let
\[
    \hat B_\tau=\sum_{j=1}^d\hat\lambda_j(\tau)\hat u_j(\tau)\hat u_j(\tau)^\top,
    \qquad
    \hat\lambda_1(\tau)\ge\cdots\ge\hat\lambda_d(\tau),
\]
be a spectral decomposition. For a candidate rank \(m\), set
\begin{equation}\label{eq:est-truncated-inverse}
    \hat B_{\tau,m}^+
    :=\sum_{j=1}^m\hat\lambda_j(\tau)^{-1}
      \hat u_j(\tau)\hat u_j(\tau)^\top,
\end{equation}
whenever \(\hat\lambda_m(\tau)>0\). If \(m=0\), set \(\hat B_{\tau,0}^+=0\). Values of \(\tau\) for which \(\hat\lambda_m(\tau)\le0\) are excluded from the minimization below. The empirical unsupervised and supervised operators are
\begin{equation}\label{eq:est-Htau}
    \hat H_{\tau,m}(v):=\hat A_\tau(v)\hat B_{\tau,m}^+,
    \qquad
    \hat H_{\tau,m}^y:=\hat A_\tau^y\hat B_{\tau,m}^+ .
\end{equation}

If \(k\) is not specified in advance, estimate it from the rank of \(M_2\). Let \(\tilde\lambda_1\ge\cdots\ge\tilde\lambda_d\) be the eigenvalues of \(\hat M_2\), and let \(a_n\) satisfy
\begin{equation}\label{eq:est-threshold-rate}
    a_n\to0,
    \qquad
    \sqrt n\,a_n\to\infty .
\end{equation}
Define
\begin{equation}\label{eq:est-khat}
    \hat k:=\sum_{j=1}^d\mathbf 1\{\tilde\lambda_j>a_n\}.
\end{equation}
When the number of topics is chosen by design or by an external model-selection step, the same formulas apply with \(\hat k\) replaced by that chosen value.

In finite samples, implementing the rank rule in \eqref{eq:est-khat} still requires a numerical choice of threshold, and nearby candidate ranks may have similar empirical second-moment spectra. When a split-sample implementation is used, as in the application below, rank choice can be made on the first split before the response-weighted moment estimation step. We use the commutator identity as a practical response-free rule of thumb: for each candidate \(k\), compute the commutator-based estimate \(\hat\alpha_0(k)\) across several admissible probe draws on the first split, and prefer ranks for which this estimate is interior and stable. This is a finite-sample rule of thumb for rank selection, not an additional rank-consistency theorem. The motivation is that, under the maintained finite-LDA model and the correct rank, admissible finite-probe commutator criteria have the same population minimizer \(\alpha_0\). Thus a working rank for which \(\hat\alpha_0(k)\) moves substantially across reasonable probe draws is less attractive than one for which the concentration estimate is interior and stable. Subsequent estimation and standard errors are computed on the independent second split, conditional on the selected rank, compression, and realized probes.

This split-sample design also allows the commutator calculation to be used a second time as a diagnostic step. Since the working rank is chosen on the first split, the commutator profile can be inspected again on the second split as a holdout diagnostic for the selected finite-dimensional specification. If the same sample were used both to choose the rank and to assess commutator stability, this second calculation would mainly repeat the selection criterion. With an independent second split, boundary behavior, flatness, multimodality, or probe sensitivity provides a separate warning about the stability of the selected spectral diagonalization. We discuss these post-selection diagnostics in Section~\ref{subsec:estimate-alpha0}.

\subsection{Estimating the concentration mass}\label{subsec:estimate-alpha0}

Let \(P(u)=I_d-uu^\top/\|u\|^2\) for \(u\ne0\). Choose raw probe vectors \(s_1,\ldots,s_L\in\R^d\), with \(L\ge2\), and define their empirical mean-orthogonal projections
\begin{equation}\label{eq:est-probes}
    \hat v_\ell:=P(\hat\mu)s_\ell,
    \qquad
    \ell=1,\ldots,L .
\end{equation}
Let \(\mathcal I\subset\{(\ell,q):1\le\ell<q\le L\}\) be a non-empty set of probe pairs. For fixed rank \(m\), define the empirical commutator criterion
\begin{equation}\label{eq:est-Q}
    \hat Q_{\mathcal I}^{(m)}(\tau)
    :=\sum_{(\ell,q)\in\mathcal I}
      \left\|
      [\hat H_{\tau,m}(\hat v_\ell),\hat H_{\tau,m}(\hat v_q)]
      \right\|_F^2,
\end{equation}
where \([R,S]=RS-SR\). If \(\hat\lambda_m(\tau)\le0\), set \(\hat Q_{\mathcal I}^{(m)}(\tau)=+\infty\). The estimator of the total concentration is
\begin{equation}\label{eq:est-alpha0}
    \hat\alpha_0
    :=\min\arg\min_{\tau\in\Theta}
      \hat Q_{\mathcal I}^{(\hat k)}(\tau),
\end{equation}
where \(\Theta=[\underline\tau,\overline\tau]\subset(0,\infty)\) is compact and contains the true \(\alpha_0\) in its interior. For the true rank \(m=k\), the finite-probe identification result in Section~\ref{sec:identification} implies that the population version of \eqref{eq:est-Q} is uniquely minimized at \(\alpha_0\), provided at least one tested pair of projected probes is non-collinear in the latent coordinates. Thus the sample criterion has the correct local target on the event \(\hat k=k\). No uniqueness claim is made here for a misspecified truncation rank.

The probability-one genericity statements are identification results, not
finite-sample guarantees. In applied work one should inspect the profile
of \(\hat Q_{\mathcal I}^{(m)}(\tau)\), whether \(\hat\alpha_0\) lies near
the boundary of \(\Theta\), local curvature around the minimizer, eigenvalue
gaps for the ordering operator, and sensitivity to the realized probe and
ordering directions. A flat, multimodal or boundary-attaining commutator
criterion, or estimates that move materially across admissible probes, should
be treated as a substantive warning. It may indicate weak finite-sample
identification, an unstable rank choice, ill-conditioned topic separation,
an inadmissible compression, or misspecification of the maintained LDA moment
structure.

\subsection{Topic and downstream coefficient estimators}\label{subsec:topic-beta-estimators}

The concentration estimator is then plugged into a single spectral operator to recover the topic columns. Choose an ordering direction \(r\in\R^d\), independent of the data if random, and set
\begin{equation}\label{eq:est-ordering-direction}
    \hat\eta:=P(\hat\mu)r .
\end{equation}
Define
\begin{equation}\label{eq:est-topic-operator}
    \hat H^{o}:=\hat H_{\hat\alpha_0,\hat k}(\hat\eta).
\end{equation}
On the high-probability event that \(\hat H^{o}\) has \(\hat k\) eigenvalues
separated from the zero cluster, and that these selected eigenvalues are real,
simple and mutually separated, let
\(\hat o_1(r),\ldots,\hat o_{\hat k}(r)\) be the associated right eigenvectors.
The selected nonzero eigenpairs are labeled by decreasing real eigenvalue within this selected cluster, and the columns are normalized by
\begin{equation}\label{eq:est-topic-normalization}
    \ones^\top\hat o_j(r)=1 .
\end{equation}
Off this event, the estimator may be defined by any fixed deterministic convention; this off-event convention is asymptotically irrelevant.
Set
\begin{equation}\label{eq:est-Ohat}
    \hat O(r):=[\hat o_1(r),\ldots,\hat o_{\hat k}(r)].
\end{equation}
The role of \(r\) is only to select a stable ordering of the topic columns. At the population value, the eigenvectors are the columns of \(O\), ordered by the entries of \(O^\top P(\mu)r\).

Finally, form the supervised operator at the estimated concentration,
\begin{equation}\label{eq:est-Hy-final}
    \hat H^y:=\hat H^y_{\hat\alpha_0,\hat k}
    =\hat A_{\hat\alpha_0}^y\hat B_{\hat\alpha_0,\hat k}^+ .
\end{equation}
The downstream coefficient estimator is
\begin{equation}\label{eq:est-beta}
    \hat\beta(r)
    :=\frac{\hat\alpha_0+2}{2}
      \diag\{\hat O(r)^+\hat H^y\hat O(r)\},
\end{equation}
where \(\diag(\cdot)\) extracts the diagonal as a vector. This is the direct sample analogue of \eqref{eq:id-beta-projection}. It does not estimate document-level topic shares and does not run a second-stage regression on generated regressors.

\subsection{Asymptotic linearity and standard errors}\label{subsec:asymptotics-inference}

The asymptotic argument is fixed-dimensional: \(d\), \(k\), and \(N\) are fixed while \(n\to\infty\). Let \(\hat{\mathcal T}_i\) denote the ordered-distinct triple tensor
\[
    \hat{\mathcal T}_i
    :=\frac{1}{N(N-1)(N-2)}
      \sum_{\substack{a,b,c=1\\ a,b,c\;\mathrm{distinct}}}^N
      x_{ia}\otimes x_{ib}\otimes x_{ic},
\]
so that \(\hat T_i(v)=\hat{\mathcal T}_i(I,I,v)\). Collect the document-level moments needed by the estimator in
\begin{equation}\label{eq:est-Zi}
    Z_i:=\{\hat\mu_i,\hat M_{2,i},\hat{\mathcal T}_i,
            Y_i,Y_i\hat\mu_i,Y_i\hat M_{2,i}\}.
\end{equation}
Let \(Z_0=\E Z_i\) and \(\bar Z=n^{-1}\sum_i Z_i\). Since the word components are bounded, the condition \(\E Y^2<\infty\) is enough for
\begin{equation}\label{eq:est-Z-clt}
    \sqrt n(\bar Z-Z_0)\rightsquigarrow N(0,\Omega),
    \qquad
    \Omega:=\Var(Z_i).
\end{equation}
On the event that the rank estimate is correct and the selected nonzero-cluster ordering eigenvalues are real, simple, mutually separated, and separated from the zero cluster, all estimators above are functions of \(\bar Z\). The commutator minimization is differentiable at the population value because the stacked commutator has a nonzero derivative with respect to \(\tau\) at \(\alpha_0\). The eigenvector and pseudoinverse maps are differentiable on the corresponding fixed-rank, separated-cluster stratum.

The following theorem records the resulting first-order expansion. If the probes or ordering direction are random, the statement is conditional on their realized values.

This inferential expansion is one of the key distinctions from the supervised spectral recovery literature:
the goal is not only consistent recovery of supervised-LDA parameters, but a first-order sampling
expansion and feasible variance estimator for the downstream coefficient.

\begin{theorem}[Feasible asymptotic inference]\label{thm:est-main-asymptotics}
Assume the model of Section~\ref{sec:model} with fixed \(d\), fixed \(k\ge3\), fixed \(N\ge3\), \(O\) full column rank, \(\alpha\in(0,\infty)^k\), and \(\E Y^2<\infty\). Let \(\Theta\) be compact with \(\alpha_0\in\mathrm{int}(\Theta)\), and let the rank threshold satisfy \eqref{eq:est-threshold-rate}. Suppose the finite-probe criterion has at least one pair \((\ell,q)\in\mathcal I\) for which \(O^\top P(\mu)s_\ell\) and \(O^\top P(\mu)s_q\) are non-collinear. Suppose also that the ordering direction \(r\) satisfies that the coordinates of \(O^\top P(\mu)r\) are distinct and nonzero.

Then
\begin{equation}\label{eq:est-consistency}
    \Pr(\hat k=k)\to1,
    \qquad
    \hat\alpha_0\xrightarrow{p}\alpha_0,
    \qquad
    \hat O(r)\xrightarrow{p}O_r,
    \qquad
    \hat\beta(r)\xrightarrow{p}\beta_r,
\end{equation}
where \(O_r\) is the topic matrix in the ordering induced by \(r\), and \(\beta_r\) is \(\beta\) in the same ordering. Moreover, there exists a mean-zero influence function \(\phi_i(r)\) with finite second moment such that
\begin{equation}\label{eq:est-joint-al}
    \sqrt n
    \begin{pmatrix}
        \hat\alpha_0-\alpha_0\\
        \vecop\{\hat O(r)-O_r\}\\
        \hat\beta(r)-\beta_r
    \end{pmatrix}
    =\frac1{\sqrt n}\sum_{i=1}^n\phi_i(r)+o_p(1).
\end{equation}
Consequently,
\begin{equation}\label{eq:est-joint-clt}
    \sqrt n
    \begin{pmatrix}
        \hat\alpha_0-\alpha_0\\
        \vecop\{\hat O(r)-O_r\}\\
        \hat\beta(r)-\beta_r
    \end{pmatrix}
    \rightsquigarrow N\{0,V(r)\},
    \qquad
    V(r):=\Var\{\phi_i(r)\}.
\end{equation}
If the probes and \(r\) are drawn independently from distributions that are absolutely continuous with respect to Lebesgue measure, the finite-probe non-collinearity condition and the distinct nonzero ordering coordinates hold with probability one over those draws. The positive spectral separation of the corrected second-moment cluster follows from the model and compactness of \(\Theta\); the ordering-operator separation from the zero cluster follows from the nonzero ordering coordinates.
\end{theorem}

\begin{corollary}[Conditional inference after split compression and controls]\label{cor:compressed-control-inference}
Under the conditions of Theorem~\ref{thm:est-main-asymptotics} and
Assumption~\ref{ass:split-compressed-controls}, consider the split-sample
compressed estimator with observed controls described in
Appendices~\ref{app:observed-controls}--\ref{app:pca-compression}. Let
\(\mathcal F_1\) denote the sigma-field generated by the first-split compression,
selected working rank, finite probes, and ordering direction, and let
\(n_2=|I_2|\to\infty\). Then, conditionally on \(\mathcal F_1\),
\begin{equation}\label{eq:est-compressed-control-corollary}
    \sqrt{n_2}\{\hat\beta^c_{\hat R_1}(r)-\beta_{\hat R_1,r}\}
    =\frac{1}{\sqrt{n_2}}\sum_{i\in I_2}\phi^c_{\hat R_1,i}(r)
      +o_p(1\mid \mathcal F_1),
\end{equation}
where
\[
    \phi^c_{\hat R_1,i}(r)
    =D^c_{\beta,r}\{Z^c_{\hat R_1,i}-Z^c_{\hat R_1,0}\},
    \qquad
    \E\{\phi^c_{\hat R_1,i}(r)\mid\mathcal F_1\}=0 .
\]
The empirical sandwich estimator obtained from the enlarged compressed control
moment vector \(Z^c_{\hat R_1,i}\) is consistent for the conditional covariance
of the right side of \eqref{eq:est-compressed-control-corollary}. The first-split
compression and first-split rank choice add no separate first-order term after
conditioning. If the conditional covariance converges in probability to a
nonrandom limit, the same normal approximation holds unconditionally.
\end{corollary}

The proof of the theorem and corollary is given in the Supplement. Its main steps are standard but useful to separate. The document-level moments satisfy the central limit theorem \eqref{eq:est-Z-clt}. The minimizer \(\hat\alpha_0\) has the usual one-dimensional minimum-distance expansion obtained by differentiating the stacked commutator equations at \(\alpha_0\). The topic estimator then follows from first-order perturbation theory for simple right eigenvectors. Finally, \eqref{eq:est-beta} is a smooth map of \((\hat\alpha_0,\hat B_{\hat\alpha_0,\hat k}^+,\hat A_{\hat\alpha_0}^y,\hat O)\); the first-order perturbation of \(\hat O\) contributes only off-diagonal terms inside \(\hat O^+\hat H^y\hat O\), so the diagonal map recovers the coefficient influence function without requiring distinct entries of \(\beta\).

For implementation, let \(G_{\beta,r}\) denote the population map from \(Z_0\) to \(\beta_r\) induced by \eqref{eq:est-alpha0}--\eqref{eq:est-beta}, with the realized probes and ordering direction held fixed. Let
\begin{equation}\label{eq:est-D-beta}
    D_{\beta,r}:=\left.\frac{\partial G_{\beta,r}(z)}{\partial z^\top}\right|_{z=Z_0} .
\end{equation}
Then the lower block of \(V(r)\) is
\begin{equation}\label{eq:est-Vbeta}
    V_\beta(r)=D_{\beta,r}\Omega D_{\beta,r}^\top .
\end{equation}
All reported standard errors use the influence-function implementation of this derivative.  Write the finite-probe commutator map as \(g(\tau,z)\), and let \(\alpha(z)\) be the local minimum-distance solution.  At an interior solution,
\begin{equation}\label{eq:est-D-alpha-impl}
    D_\alpha
    :=\left.\frac{\partial \alpha(z)}{\partial z^\top}\right|_{z=Z_0}
    =-\{G_\tau^\top G_\tau\}^{-1}G_\tau^\top G_z,
    \qquad
    G_\tau=\left.\frac{\partial g(\tau,z)}{\partial \tau}\right|_{(\alpha_0,Z_0)},\quad
    G_z=\left.\frac{\partial g(\tau,z)}{\partial z^\top}\right|_{(\alpha_0,Z_0)} .
\end{equation}
For fixed \(\alpha_0\), the derivative of the coefficient map is obtained by differentiating the corrected moment operators, the truncated inverse of the corrected second moment, the normalized right eigenvectors used to form \(\hat O(r)\), and the final diagonal map in \eqref{eq:est-beta}.  The full derivative has the chain-rule form
\begin{equation}\label{eq:est-D-beta-chain}
    D_{\beta,r}=D_{\beta,r}^{\rm fix}+D_{\beta,\alpha}D_\alpha,
\end{equation}
where \(D_{\beta,r}^{\rm fix}\) is the derivative with \(\alpha_0\) held fixed and \(D_{\beta,\alpha}\) is the derivative of the coefficient map with respect to \(\alpha_0\).  The empirical influence values are therefore
\begin{equation}\label{eq:est-emp-if-beta}
    \hat\phi_{\beta,i}(r)=\hat D_{\beta,r}\{Z_i-\bar Z\} .
\end{equation}
Equivalently, with
\begin{equation}\label{eq:est-sandwich}
    \hat\Omega:=\frac1n\sum_{i=1}^n\{Z_i-\bar Z\}\{Z_i-\bar Z\}^\top,
    \qquad
    \hat V_\beta(r)=\hat D_{\beta,r}\hat\Omega\hat D_{\beta,r}^\top
    =\frac1n\sum_{i=1}^n\hat\phi_{\beta,i}(r)\hat\phi_{\beta,i}(r)^\top,
\end{equation}
where \(\hat D_{\beta,r}\) is the sample analogue of \eqref{eq:est-D-beta-chain}.  Equations~\eqref{eq:est-emp-if-beta}--\eqref{eq:est-sandwich} describe the no-control moment vector. If observed controls are included, the same formulas are applied after replacing \(Z_i\), \(\bar Z\), \(D_{\beta,r}\), and \(\Omega\) by the enlarged control moment vector \(Z_i^c\), its mean, the derivative of the control-adjusted map, and \(\Var(Z_i^c)\). If the data are first compressed, \(Z_i^c\) is formed from the compressed tokens and includes the compressed scale-recovery moments used by the control adjustment. A pointwise \((1-a)\) confidence interval for the \(j\)-th ordered coefficient is
\begin{equation}\label{eq:est-ci}
    \hat\beta_j(r)
    \pm
    z_{1-a/2}\left\{\frac{\hat V_{\beta,jj}(r)}{n}\right\}^{1/2}.
\end{equation}
The same covariance estimate is used for Wald tests and linear contrasts of the ordered coefficient vector. For pointwise confidence intervals we assume \(V_{\beta,jj}(r)>0\). For Wald tests of \(C\beta=c\), we assume \(CV_\beta(r)C^\top\) is nonsingular.

\section{Monte Carlo experiments}\label{sec:simulation}

This section evaluates the finite-sample behavior of the proposed estimator in the
setting for which the theory is designed: downstream regression with a fixed number
of words per document.  The experiments are organized around two questions.  First,
does the direct moment-based estimator deliver calibrated inference for the
coefficient vector \(\beta\)?  Second, what goes wrong if the analyst instead
constructs document-level topic estimates and treats them as observed regressors?
The number of topics is supplied to all methods, so the exercise isolates the
inferential problem rather than the separate rank-selection step.

To make the fixed-document plug-in issue explicit, consider any deterministic reconstruction rule
\(\tilde h_i=g_N(c_i;O)\) based on the document counts and the known population topic matrix.
If a no-intercept population least-squares regression of \(Y_i\) on \(\tilde h_i\) is nonsingular and
\(\E(\tilde h_i\varepsilon_i)=0\), its probability limit is
\begin{equation}\label{eq:plugin-plim}
    \beta_g
    =\{\E(\tilde h_i\tilde h_i^\top)\}^{-1}\E(\tilde h_i h_i^\top)\beta .
\end{equation}
Thus the plug-in coefficient equals \(\beta\) only under additional restrictions, such as
\(\E(\tilde h_i h_i^\top)=\E(\tilde h_i\tilde h_i^\top)\). The correctly specified posterior mean
\(\tilde h_i=\E(h_i\mid c_i)\) is a special case satisfying this identity by iterated expectations, but a generic finite-\(N\) reconstruction rule, including the constrained least-squares rule below, need not satisfy it. The simulations therefore use the known-\(O\) plug-in regression as an analytic and numerical benchmark for the persistent generated-regressor problem at fixed document length.

\subsection{Design and methods}

The vocabulary size is \(d=100\), the number of topics is \(k=10\), and each
document contains \(N=100\) words.  The topic matrix \(O\) is generated once and
then held fixed across Monte Carlo replications.  Its columns are independently
generated from a Dirichlet distribution with concentration parameter \(0.3\),
subject to screening restrictions that remove nearly duplicate topics and nearly
anchor-word designs.  For the realized matrix, the minimum pairwise Hellinger
distance between topic columns is \(0.588\), the maximum row dominance is
\(0.879\), the minimum Hellinger distance from any row-normalized word profile to
a simplex vertex is \(0.250\), and there are no exact zeros.

For document \(i\),
\[
    h_i\sim\operatorname{Dirichlet}(\alpha),\qquad
    c_i\mid h_i\sim\operatorname{Multinomial}(N,Oh_i),
\]
and the response is
\[
    Y_i=\beta^\top h_i+\varepsilon_i,
    \qquad
    \varepsilon_i\sim N(0,\sigma^2),
    \qquad
    \beta=(1.0,0.9,\ldots,0.1)^\top .
\]
The noise variance is calibrated to give population \(R^2=0.35\).  We consider a
symmetric Dirichlet design, \(\alpha_j=0.5\) for \(j=1,\ldots,10\), so that
\(\alpha_0=5\), and an asymmetric design, \(\alpha_j=2j/55\), so that
\(\alpha_0=2\).  The main downstream comparison uses \(n\in\{3000,5000,15000\}\); each design cell uses \(1000\) Monte Carlo replications.  The concentration-parameter table below also reports the original \(n=30000\) design point.

We compare four procedures.  The first is an infeasible no-intercept latent-oracle regression
that uses the true \(h_i\)'s.  The second is the proposed direct spectral estimator,
which estimates \(\alpha_0\), recovers the topic matrix, and estimates \(\beta\)
from the supervised corrected moments.  Its confidence intervals use the
influence-function implementation of the first-order variance formula developed
above: the derivative of \(\hat\alpha_0\) is obtained from the implicit
minimum-distance expansion of the commutator criterion, and the derivative of
\(\hat\beta\) is the fixed-\(\alpha_0\) spectral derivative plus the chain-rule
contribution from estimating \(\alpha_0\).  The third
procedure is a known-\(O\) plug-in regression: it uses the true topic matrix \(O\), but
estimates document-level topic mixtures by
\[
    \tilde h_i
    =\arg\min_{h\in\Delta^{k-1}}
    \left\|\frac{c_i}{N}-Oh\right\|_2^2,
\]
and then applies no-intercept ordinary least squares treating \(\tilde h_i\) as
observed.  The no-intercept specification is natural because
\(\mathbf 1^\top h_i=1\).  The fourth procedure is the analogous spectral plug-in regression, replacing \(O\) by
the unsupervised spectral estimate \(\hat O\).  Thus the final two procedures
mimic the modular workflow used in many empirical applications.  Topic labels are
aligned to the true labels by Hungarian matching only when computing simulation
errors and coverage.

For each method, Table~\ref{tab:sim-beta} reports the root mean squared
error
\[
    \operatorname{RMSE}_\beta
    =
    \left\{
    \frac{1}{1000}\sum_{r=1}^{1000}
    \frac{1}{k}\|\hat\beta^{(r)}-\beta\|_2^2
    \right\}^{1/2},
\]
the average coordinate-wise coverage of nominal \(95\%\) confidence intervals, the
minimum coordinate-wise coverage, average interval length, and average Hellinger
distance between matched topic columns.

\begin{table}[!htbp]
\centering
\small
\setlength{\tabcolsep}{4.5pt}
\begin{threeparttable}
\caption{Downstream inference at fixed document length}
\label{tab:sim-beta}
\begin{tabular}{@{}llrrrrr@{}}
\toprule
\(n\) & Method
& \(\operatorname{RMSE}_\beta\)
& Cov.
& Min cov.
& Length
& Topic \(H\) \\
\midrule
\multicolumn{7}{@{}l}{\textit{Panel A: symmetric Dirichlet, \(\alpha_0=5\)}} \\
3,000 & Latent oracle              & 0.022 & 0.949 & 0.935 & 0.085 & -- \\
3,000 & Direct spectral            & 0.073 & 0.978 & 0.957 & 0.452 & 0.205 \\
3,000 & Plug-in, true \(O\)        & 0.046 & 0.607 & 0.126 & 0.086 & 0.000 \\
3,000 & Plug-in, \(\hat O\)        & 0.068 & 0.529 & 0.270 & 0.088 & 0.205 \\
\addlinespace
5,000 & Latent oracle              & 0.017 & 0.954 & 0.943 & 0.066 & -- \\
5,000 & Direct spectral            & 0.058 & 0.969 & 0.944 & 0.331 & 0.178 \\
5,000 & Plug-in, true \(O\)        & 0.044 & 0.492 & 0.016 & 0.067 & 0.000 \\
5,000 & Plug-in, \(\hat O\)        & 0.063 & 0.461 & 0.226 & 0.068 & 0.178 \\
\addlinespace
15,000 & Latent oracle             & 0.010 & 0.951 & 0.934 & 0.038 & -- \\
15,000 & Direct spectral           & 0.027 & 0.962 & 0.937 & 0.113 & 0.125 \\
15,000 & Plug-in, true \(O\)       & 0.042 & 0.272 & 0.000 & 0.039 & 0.000 \\
15,000 & Plug-in, \(\hat O\)       & 0.052 & 0.288 & 0.028 & 0.038 & 0.125 \\
\addlinespace[1mm]
\multicolumn{7}{@{}l}{\textit{Panel B: asymmetric Dirichlet, \(\alpha_0=2\)}} \\
3,000 & Latent oracle              & 0.024 & 0.950 & 0.940 & 0.087 & -- \\
3,000 & Direct spectral            & 0.072 & 0.966 & 0.948 & 0.447 & 0.168 \\
3,000 & Plug-in, true \(O\)        & 0.045 & 0.818 & 0.422 & 0.091 & 0.000 \\
3,000 & Plug-in, \(\hat O\)        & 0.079 & 0.647 & 0.328 & 0.091 & 0.168 \\
\addlinespace
5,000 & Latent oracle              & 0.018 & 0.954 & 0.938 & 0.067 & -- \\
5,000 & Direct spectral            & 0.056 & 0.967 & 0.947 & 0.300 & 0.145 \\
5,000 & Plug-in, true \(O\)        & 0.042 & 0.746 & 0.168 & 0.070 & 0.000 \\
5,000 & Plug-in, \(\hat O\)        & 0.069 & 0.613 & 0.252 & 0.070 & 0.145 \\
\addlinespace
15,000 & Latent oracle             & 0.011 & 0.952 & 0.936 & 0.039 & -- \\
15,000 & Direct spectral           & 0.029 & 0.962 & 0.943 & 0.141 & 0.103 \\
15,000 & Plug-in, true \(O\)       & 0.039 & 0.506 & 0.000 & 0.041 & 0.000 \\
15,000 & Plug-in, \(\hat O\)       & 0.051 & 0.460 & 0.059 & 0.040 & 0.103 \\
\bottomrule
\end{tabular}
\begin{tablenotes}
\footnotesize
\item[] \textit{Notes:} Cov. is the average coordinate-wise coverage of nominal
\(95\%\) confidence intervals; Min cov. is the minimum over the ten coordinates.
Length is the average coordinate-wise interval length.  Topic \(H\) is the
average Hellinger distance between matched topic columns; it is not applicable for
the latent-oracle regression and equals zero for methods using the true topic
matrix.
\end{tablenotes}
\end{threeparttable}
\end{table}

The contrast between the direct estimator and the plug-in regressions is the main
finding.  In the symmetric design, the direct spectral estimator has average
coverage \(0.978\), \(0.969\), and \(0.962\) as \(n\) increases from \(3000\) to
\(15000\).  Over the same sample sizes, the spectral plug-in regression has
coverage \(0.529\), \(0.461\), and \(0.288\).  The known-\(O\) plug-in regression,
which is given the true topic matrix, also undercovers severely: its coverage
falls from \(0.607\) to \(0.272\).  Thus the failure is not primarily a topic
recovery error.  It is the fixed-document generated-regressor error induced by
estimating \(h_i\) from only \(N=100\) words and then treating \(\tilde h_i\) as
observed.

The asymmetric design gives the same qualitative conclusion.  The direct spectral
estimator remains close to nominal coverage, whereas the true-topic plug-in
coverage falls from \(0.818\) to \(0.506\), and the spectral plug-in coverage
falls from \(0.647\) to \(0.460\).  The plug-in intervals have lengths comparable
to the latent-oracle intervals, so they shrink at the usual \(n^{-1/2}\) rate even
though their centers retain finite-document bias.  With \(1000\) replications, the
Monte Carlo standard error of a single coordinate's nominal \(95\%\) coverage
estimate is about \(\{0.95(0.05)/1000\}^{1/2}=0.0069\), far smaller than the
reported plug-in undercoverage.

As a further diagnostic, we reran the spectral plug-in regression while giving it
the true value of \(\alpha_0\).  Across the six design cells in
Table~\ref{tab:sim-beta}, this changed average coverage by at most
\(0.039\).  The plug-in failure is therefore not explained by uncertainty in the
concentration parameter; it comes from using finite-document topic estimates as
regressors.

\subsection{Estimation of the concentration parameter}

Table~\ref{tab:sim-alpha} reports the behavior of the commutator estimator of
\(\alpha_0\).  Bias and RMSE decrease steadily with \(n\) in both Dirichlet
designs.  At the two largest sample sizes, the empirical standard deviation and
the average estimated standard error are close, supporting the first-order
linearization used in the downstream variance formula.  Together with
Table~\ref{tab:sim-beta}, this indicates that the proposed method learns
\(\alpha_0\) from word moments and delivers calibrated downstream inference
without plug-in document-topic regressors.

\begin{table}[!htbp]
\centering
\scriptsize
\setlength{\tabcolsep}{5pt}
\begin{threeparttable}
\caption{Estimation of \(\alpha_0\)}
\label{tab:sim-alpha}
\begin{tabular}{@{}lrrrrrrr@{}}
\toprule
\(n\)
& Mean
& Bias
& RMSE
& Cov.
& Emp. sd
& Mean se
& se/sd \\
\midrule
\multicolumn{8}{@{}l}{\textit{Panel A: symmetric Dirichlet, \(\alpha_0=5\)}} \\
3,000  & 5.984 & 0.984 & 1.364 & 0.961 & 0.944 & 1.109 & 1.174 \\
5,000  & 5.634 & 0.634 & 1.000 & 0.949 & 0.774 & 0.812 & 1.050 \\
15,000 & 5.209 & 0.209 & 0.472 & 0.951 & 0.423 & 0.429 & 1.012 \\
30,000 & 5.112 & 0.112 & 0.322 & 0.941 & 0.302 & 0.297 & 0.984 \\
\addlinespace
\multicolumn{8}{@{}l}{\textit{Panel B: asymmetric Dirichlet, \(\alpha_0=2\)}} \\
3,000  & 2.284 & 0.284 & 0.413 & 0.928 & 0.300 & 0.324 & 1.080 \\
5,000  & 2.167 & 0.167 & 0.273 & 0.940 & 0.216 & 0.236 & 1.095 \\
15,000 & 2.054 & 0.054 & 0.134 & 0.951 & 0.123 & 0.127 & 1.034 \\
30,000 & 2.028 & 0.028 & 0.091 & 0.948 & 0.087 & 0.088 & 1.018 \\
\bottomrule
\end{tabular}
\begin{tablenotes}
\footnotesize
\item[] \textit{Notes:} Results are for the proposed direct spectral estimator.
Cov. is the empirical coverage of the nominal \(95\%\) confidence interval for
\(\alpha_0\).  The last column is the ratio of the average estimated standard
error to the Monte Carlo standard deviation of \(\hat\alpha_0\).
\end{tablenotes}
\end{threeparttable}
\end{table}

The simulation designs deliberately hold fixed the topic dimension, topic
separation, document length, concentration search interval and probe
construction. They are therefore best read as tests of the fixed-document
generated-regressor mechanism and of the first-order variance formula in
well-separated designs. They do not replace application-specific diagnostics
for the commutator criterion, ordering eigengaps, rank selection or
compression stability.

\section{Application: abstract topics and citations in economics}\label{sec:application}

We illustrate the method on articles published in the five general-interest economics journals
\emph{American Economic Review}, \emph{Econometrica}, \emph{Journal of Political Economy},
\emph{Quarterly Journal of Economics}, and \emph{Review of Economic Studies}. Article metadata and
citation counts are taken from OpenAlex \citep{PriemPiwowarOrr2022,OpenAlexData2026}. The response is
\[
        Y_i=\log\{1+\text{citation count}_i\},
\]
and the document consists of the article title and abstract. The analysis is descriptive: a topic
coefficient is interpreted as an association between latent abstract content and citation impact,
not as a causal effect of writing on that topic. The exercise should also be read under the maintained
low-order response-token orthogonality restrictions, which cannot be verified from the observed text alone.

The maintained response equation in this application is
\[
        Y_i = \beta^\top h_i + \delta^\top q_i + \varepsilon_i,
\]
where \(h_i\) is the latent topic mixture and \(q_i\) contains publication-year indicators. Since
\(\ones^\top h_i=1\), a common shift in all topic coefficients is observationally equivalent to an
intercept shift. Thus the absolute level of an individual coefficient \(\beta_j\), and a test of
\(\beta_j=0\), depends on the normalization implicit in the intercept. The invariant objects are
topic contrasts,
\[
        \beta_a-\beta_b,
\]
which compare the conditional mean response when latent topic mass is shifted from topic \(b\) to
topic \(a\). A reallocation of 10 percentage points of topic mass from \(b\) to \(a\) corresponds to
\(0.1(\beta_a-\beta_b)\) log-points in \(Y_i\), holding the year controls fixed. We therefore report
the topic coefficient estimates to orient the reader, but interpret the empirical results through
joint and pairwise contrast tests.

The raw corpus contains records published between 2000 and 2018. We exclude records whose titles or
metadata indicate comments, replies, corrections, editorial material, or very short documents. We
lowercase tokens, remove standard English and scholarly stopwords, and keep terms whose document
frequency is at least 40. The main specification does not stem words. The resulting analysis sample
has 5409 documents and 1771 terms. Because the vocabulary dimension is too large for direct dense
third-order moment calculations, we use the split-sample compression described in
Appendix~\ref{app:pca-compression}. A 10\% split of documents is used to estimate a 100-dimensional
orthonormal compression from the uncentered cross-token second moment, and the remaining 90\% is
used for estimation and inference. All reported confidence intervals and standard errors are computed
on the second split, conditional on the first-split compression. This follows the split-compression
construction in Appendix~\ref{app:pca-compression}, under which the first-stage compression is held
fixed for second-split inference. The empirical moments use document-specific length normalizations
as in Appendix~\ref{app:variable-lengths}.

Following the split-sample finite-rank diagnostic described in Section~\ref{subsec:plugin-operators}, the working topic dimension is selected using only first-split word moments, without using citation outcomes. For each candidate \(k\in\{10,\ldots,20\}\), we estimate \(\alpha_0\) from the commutator criterion across 20 independent sets of probe directions, using five mean-orthogonal commutator partner directions for each run. Candidates with boundary or failed concentration estimates are discarded. Among the remaining candidates, we choose the value of \(k\) with the smallest standard deviation of \(\log \hat\alpha_0\). This first-split commutator-stability diagnostic selects \(k=15\).

After fixing this rank and the first-split compression, we run the estimator on the independent second split and inspect the second-split commutator profile, boundary behavior, compressed scale recovery, and control-adjustment conditioning as diagnostics for the selected specification. In a second-split probe diagnostic using five admissible probe draws, \(\hat\alpha_0\) ranges from 5.09 to 7.59, and none of the corresponding commutator profiles has a boundary minimum. The second-split estimator uses year fixed effects, with year 2000 as the omitted category. The coefficient estimates and standard errors are computed in compressed coordinates after the moment-based topic-scale recovery used in the observed-control adjustment; lifted word-space directions are used only to label topics. The selected specification has
\[
        \hat\alpha_0=5.09,
        \qquad
        n_{\rm est}=4868,
        \qquad
        \operatorname{cond}(\hat\Gamma_q)=6.56,
\]
where \(\hat\Gamma_q\) is the matrix in the observed-control adjustment. The minimum absolute recovered scale factor is 1.78, and the value of \(\operatorname{cond}(\hat\Gamma_q)\) is moderate in this application, so the year-control adjustment is not close to singular by this diagnostic.

Table~\ref{tab:top5-main} reports the topic labels and centered coefficient estimates in the selected ordering. The labels are assigned from the leading positive words of the lifted word-space directions. The centered estimates subtract the average topic coefficient and are shown only to orient the reader; the invariant inferential objects are the joint equality test and pairwise contrasts reported below.

\begin{table}[H]
\centering
\begin{threeparttable}
\caption{Topic labels and centered coefficients for top-five economics articles}
\label{tab:top5-main}
\scriptsize
\setlength{\tabcolsep}{3.6pt}
\begin{tabular}{r>{\raggedright\arraybackslash}p{0.24\linewidth}>{\raggedright\arraybackslash}p{0.38\linewidth}rr}
\toprule
Topic & Label & Leading words & Centered est. & S.e. \\
\midrule
1  & Risk / insurance / asset preferences
   & risk, insurance, financial, market, sharing, utility
   & 0.317 & 0.364 \\
2  & Capital / growth / inequality
   & capital, growth, tax, model, market, income
   & 0.035 & 0.289 \\
3  & Tax / income / health / public finance
   & tax, income, taxes, consumption, health, reform
   & -0.016 & 0.327 \\
4  & Macro / monetary / shocks
   & model, rate, policy, monetary, shocks, inflation
   & 1.096 & 0.216 \\
5  & Markets / policy / welfare
   & market, policy, price, markets, optimal, welfare
   & -0.124 & 0.244 \\
6  & Monetary policy / inflation
   & policy, monetary, rate, inflation, growth, shocks
   & 0.442 & 0.273 \\
7  & Labor search / wages
   & market, model, price, labor, workers, search
   & 0.138 & 0.384 \\
8  & Consumer choice / prices / demand
   & price, prices, choice, preferences, utility, consumers
   & 0.287 & 0.422 \\
9  & Auctions / prices / products
   & price, prices, auction, auctions, firms, demand
   & -0.381 & 0.189 \\
10 & Trade / international
   & trade, countries, free, country, international, growth
   & 0.563 & 2.203 \\
11 & Game theory / auctions / equilibrium
   & equilibrium, games, trade, game, players, auctions
   & -1.334 & 2.381 \\
12 & Mechanism design / principal--agent
   & optimal, agent, principal, mechanisms, contracts, risk
   & -0.750 & 0.904 \\
13 & Firms / trade / productivity
   & firms, firm, model, trade, workers, productivity
   & 0.962 & 0.274 \\
14 & Private information / mechanism design
   & information, private, agents, agent, models, equilibrium
   & -0.460 & 0.302 \\
15 & Econometrics / identification / asymptotics
   & models, estimator, estimators, estimation, asymptotic, identification
   & -0.773 & 0.230 \\
\bottomrule
\end{tabular}
\begin{tablenotes}[flushleft]
\scriptsize
\item Notes: The response is \(\log\{1+\text{citation count}\}\). Documents are title and abstract text.
The sample contains top-five economics articles from 2000--2018. The estimator uses a 10\%/90\%
split-sample PCA compression with \(m=100\), first-split commutator-stability selection of \(k\)
using word moments only, and year fixed effects. The coefficient estimates and standard errors are computed in compressed coordinates after moment-based topic-scale recovery. The displayed estimates are \(\hat\beta_j-\bar{\hat\beta}\), with standard errors computed from the full analytic sandwich covariance. Centering removes the common level that absorbs the intercept normalization because topic shares sum to one; pairwise contrasts are unchanged by this centering.
\end{tablenotes}
\end{threeparttable}
\end{table}

The topics align with recognizable fields and methods in economics. The recovered topic labeled
``Econometrics / identification / asymptotics'', for example, has leading words
\emph{models}, \emph{estimator}, \emph{estimators}, \emph{estimation}, \emph{asymptotic},
and \emph{identification}. Other topics correspond to macroeconomic shocks, monetary policy,
labor search, public finance, firms and productivity, auctions, mechanism design, and private
information. This is a qualitative check that the spectral directions are not merely recovering common
article boilerplate.

The natural omnibus null in this setting is not \(\beta=0\), but equality of all topic coefficients,
\[
        H_0:\beta_1=\cdots=\beta_{15}.
\]
This null says that, after controlling for publication year, latent topic composition has no association
with citation impact. Using the full influence-function covariance matrix, the Wald statistic for this
14-dimensional restriction is 56.26, with \(p=5.25\times10^{-7}\). Thus, within the maintained
latent-topic model, abstract topic composition is associated with citation impact.

Table~\ref{tab:top5-contrasts} reports the pairwise contrasts that remain significant after Holm
adjustment over all 105 pairwise comparisons. The largest contrast is between the macro/monetary
shocks topic and the econometrics/identification topic. A ten-percentage-point shift of latent topic
mass from the latter to the former corresponds to 0.187 log-points in the citation outcome, holding
publication year fixed. The firms/trade/productivity topic also has a larger citation association than
the econometrics/identification topic and the auctions/prices/products topic. These are descriptive
contrasts in latent topic coordinates, not causal comparisons across research fields.

\begin{table}[H]
\centering
\begin{threeparttable}
\caption{Holm-significant topic contrasts}
\label{tab:top5-contrasts}
\scriptsize
\setlength{\tabcolsep}{3.0pt}
\begin{tabular}{>{\raggedright\arraybackslash}p{0.255\linewidth}>{\raggedright\arraybackslash}p{0.285\linewidth}rrrr}
\toprule
Higher topic & Lower topic & Difference & S.e. & \(p\) & Holm \(p\) \\
\midrule
Macro / monetary / shocks
  & Econometrics / identification / asymptotics
  & 1.868 & 0.323 & \(7.6\times10^{-9}\) & \(8.0\times10^{-7}\) \\
Macro / monetary / shocks
  & Auctions / prices / products
  & 1.477 & 0.292 & \(4.2\times10^{-7}\) & \(4.36\times10^{-5}\) \\
Firms / trade / productivity
  & Econometrics / identification / asymptotics
  & 1.734 & 0.364 & \(2.0\times10^{-6}\) & \(2.01\times10^{-4}\) \\
Macro / monetary / shocks
  & Private information / mechanism design
  & 1.556 & 0.394 & \(8.0\times10^{-5}\) & 0.0082 \\
Macro / monetary / shocks
  & Markets / policy / welfare
  & 1.220 & 0.310 & \(8.2\times10^{-5}\) & 0.0083 \\
Firms / trade / productivity
  & Auctions / prices / products
  & 1.343 & 0.356 & \(1.6\times10^{-4}\) & 0.016 \\
\bottomrule
\end{tabular}
\begin{tablenotes}[flushleft]
\scriptsize
\item Notes: Contrasts are differences between the corresponding topic coefficients in
Table~\ref{tab:top5-main}. Holm \(p\)-values adjust over all \(15\times14/2=105\) pairwise topic
comparisons. Standard errors use the full analytic sandwich covariance, including the contribution
from estimating \(\alpha_0\) and the compressed scale-recovery step used for the observed-control adjustment.
\end{tablenotes}
\end{threeparttable}
\end{table}

The direction of these contrasts is consistent with existing evidence on citation heterogeneity across
fields in economics. \citet{CardDellaVigna2013} study the same five general-interest journals and
document substantial differences in citation patterns across fields, including relatively high citation
performance for more recent Development and International Economics papers and lower relative
citation performance for more recent Econometrics and Theory papers. \citet{AnauatiGalianiGalvez2016}
classify top-five articles into applied, applied theory, econometric methods, and theory, and find
that citation life cycles differ markedly across these categories: applied and applied-theory papers
have more favorable citation profiles, while theory and econometric-method papers have lower
typical citation profiles, with econometric methods displaying a more heterogeneous upper tail. Our
topic-level estimates should not be read as a replication of those field classifications, but the broad
pattern is similar: topic content is associated with citation outcomes, and the more applied macro and
firm/trade/productivity directions have higher citation associations than some econometric and
theory-oriented directions.

The application is intended as a numerical illustration of the moment-based inferential workflow. It
uses a single first-split selection rule, fixed before estimating the second-split coefficients, and all
reported uncertainty is computed on the independent second split, conditional on the first-split
compression and selected working dimension. Appendix~\ref{app:application-robustness} reports second-split diagnostics, adjacent-rank checks, and a preprocessing robustness check. The main text is kept to one specification to avoid turning the illustration into a model-selection exercise.

\section{Discussion}\label{sec:discussion}

This paper develops direct inference for regression with latent Dirichlet
topic covariates. The main distinction from the usual plug-in workflow is
that document-level topic shares are not estimated and then treated as
observed regressors. At fixed document length, those document-level
quantities remain noisy even when the topic matrix is known. The proposed
response-weighted corrected moments instead identify the downstream
coefficient \(\beta\) directly from observable word--response moments under
explicit low-order response-token orthogonality conditions. The same
corrected moment structure also yields an estimating equation for the
unknown total concentration parameter \(\alpha_0\), through the
commutativity of corrected word-moment operators.

The results are model-based and should be interpreted within the stated
asymptotic regime. The formal theory keeps the vocabulary dimension, number
of topics and document length fixed while the number of documents increases.
The commutativity identification of \(\alpha_0\) is formulated for
\(k\ge3\) and relies on rank, separation and generic probe conditions.
These genericity conditions should not be read as a substitute for numerical
diagnostics in applied use. Boundary solutions, flat commutator profiles,
small ordering eigengaps, unstable rank choices or substantial sensitivity to
admissible probes are substantive warnings. They may reflect weak finite-sample
identification, ill-conditioned topic separation, an unstable compression, or
misspecification of the maintained LDA moment structure. The split-compression
results in the Supplement identify \(\beta\) conditional on an admissible
projection; compressed directions are not, by themselves, simplex-normalized
topics.

Integrated latent-variable likelihood or Bayesian joint-model procedures
could also be used to target \(\beta\) under fuller specifications by
integrating over the latent document mixtures. The contribution here is a
different tradeoff: low-order moment-based frequentist inference for
\(\beta\), without document-level topic-share regression and without a full
likelihood or prior specification for all unknown population parameters.

Several extensions remain. A theory allowing the vocabulary dimension or
the compression dimension to grow with the number of documents would better
match large-scale text applications. It would also be useful to study
robustness under departures from the Dirichlet topic-mixture assumption,
including non-Dirichlet latent mixtures, correlated topics and misspecified
topic dimension. Finally, the same moment logic may be useful beyond topic
models, in other latent simplex models where the inferential target is a
downstream regression coefficient rather than recovery of each unit's latent
coordinate.

\newpage
\appendix

\section{Supplementary material}\label{app:supplement}

This supplement gives the details behind Sections~\ref{sec:model}--\ref{sec:estimation} and records application diagnostics and robustness checks for Section~\ref{sec:application}. All limits are taken with fixed vocabulary size \(d\), fixed number of topics \(k\), and fixed document length \(N\ge3\), while the number of documents \(n\) diverges. If probe vectors or the ordering direction are generated at random, probability statements in the asymptotic theory are conditional on their realized values, except where the full-measure properties of such random draws are stated explicitly.

For a vector \(a\), \(\diag(a)\) is the diagonal matrix with diagonal \(a\). For a square matrix \(A\), \(\diag(A)\) is the vector formed from the diagonal of \(A\). The Moore--Penrose inverse of a full-column-rank matrix \(O\) is \(O^+=(O^\top O)^{-1}O^\top\). The commutator is \([A,B]=AB-BA\). Constants denoted by \(C\) are finite and may change from line to line.

\subsection{Application diagnostics and robustness}\label{app:application-robustness}

This subsection records the finite-sample diagnostics and robustness checks for
Section~\ref{sec:application}. All entries use the moment-based compressed scale
recovery in Proposition~\ref{prop:app-compressed-controls}. The main specification
is the non-stemmed document-frequency-40 corpus, the 10\%/90\% split-sample PCA
compression with \(m=100\), year fixed effects, and the first-split selected rank
\(k=15\).

Table~\ref{tab:app-diagnostics} gives second-split diagnostics for the selected
specification. The selected probe draw gives an interior concentration estimate, a
moderate control-adjustment condition number, and no near-zero compressed scale
factor. Repeating the second-split commutator calculation over five admissible
probe draws gives concentration estimates between 5.09 and 7.59, with no boundary
minima.

\begin{table}[!htbp]
\centering
\begin{threeparttable}
\caption{Second-split diagnostics for the main application specification}
\label{tab:app-diagnostics}
\scriptsize
\setlength{\tabcolsep}{4.0pt}
\begin{tabular}{@{}lrrrrrr@{}}
\toprule
Specification & \(\hat\alpha_0\) & Probe range & Boundary hits & \(\operatorname{cond}(\hat\Gamma_q)\) & Min. scale & Omnibus \(p\) \\
\midrule
Main \(k=15\) & 5.09 & 5.09--7.59 & 0/5 & 6.56 & 1.78 & \(5.25\times10^{-7}\) \\
\bottomrule
\end{tabular}
\begin{tablenotes}[flushleft]
\scriptsize
\item Notes: The probe range is the range of \(\hat\alpha_0\) over five second-split
commutator calculations with different admissible probe draws. ``Min. scale'' is the
minimum absolute compressed scale-recovery factor in the selected full-inference run.
\end{tablenotes}
\end{threeparttable}
\end{table}

Table~\ref{tab:app-adjacent-rank} checks adjacent ranks using the same corpus,
compression, controls, probe seed, and moment-scale recovery as the main specification.
The selected rank \(k=15\) gives the sharpest pairwise contrast evidence, but the
neighboring ranks also reject the omnibus null that all topic coefficients are equal.
Thus the evidence that topic composition is associated with citation outcomes is not
confined to a single rank choice.

\begin{table}[!htbp]
\centering
\begin{threeparttable}
\caption{Adjacent-rank robustness for the application}
\label{tab:app-adjacent-rank}
\scriptsize
\setlength{\tabcolsep}{4.0pt}
\begin{tabular}{@{}rrrrrrrr@{}}
\toprule
\(k\) & \(\hat\alpha_0\) & s.e.\((\hat\alpha_0)\) & \(\operatorname{cond}(\hat\Gamma_q)\) & Min. scale & Omnibus \(p\) & Nominal pairs & Holm pairs \\
\midrule
14 & 5.67 & 0.58 & 4.76 & 0.73 & \(1.33\times10^{-4}\) & 23 & 1 \\
15 & 5.09 & 0.55 & 6.56 & 1.78 & \(5.25\times10^{-7}\) & 21 & 6 \\
16 & 6.94 & 0.72 & 4.97 & 0.32 & \(6.21\times10^{-4}\) & 14 & 2 \\
\bottomrule
\end{tabular}
\begin{tablenotes}[flushleft]
\scriptsize
\item Notes: ``Nominal pairs'' is the number of pairwise topic contrasts significant at nominal 5\%.
``Holm pairs'' is the number that remain significant after Holm adjustment over all pairwise
contrasts for the given rank. All runs use year fixed effects and the moment-scale control adjustment.
\end{tablenotes}
\end{threeparttable}
\end{table}

The Holm-significant adjacent-rank contrasts are also directionally consistent with the
main display. For \(k=14\), the firms/trade/productivity topic is larger than a
private-information/mechanism topic. For \(k=16\), the firms/trade/productivity topic and an
education/health/empirical-effects topic are both larger than the econometrics/identification/asymptotics
topic. The exact pairwise contrast table changes with the working rank, as expected, but applied
macro, firm, trade and empirical directions remain the high-association families relative to several
methodological or theory-oriented directions.

As a preprocessing robustness check, we repeated the analysis using a light-stemmed vocabulary, a
30-document-frequency threshold, and the same 10\%/90\% split-compression design. The first-split
commutator-stability diagnostic selected \(k=17\). This alternative specification gives
\(\hat\alpha_0=6.96\), \(\operatorname{cond}(\hat\Gamma_q)=3.90\), a minimum absolute scale factor
of 0.52, and an omnibus test \(p\)-value of 0.0045. It rejects equality of all topic coefficients and
recovers many broad economics families, including public finance, macro policy, labor, firms and
productivity, trade, mechanism-design or contract topics, risk or asset themes, auctions or demand,
and econometrics. Its pairwise contrasts are less sharp after multiplicity adjustment, with no
Holm-significant pairwise contrast. We therefore use the non-stemmed document-frequency-40
specification as the main display.

\subsection{Assumptions used in the asymptotic theory}\label{app:assumptions}

The main theorem in Section~\ref{sec:estimation} is stated in compact form. The following assumptions record the regularity conditions used in the proof.

\begin{assumption}[Finite-LDA model and response]\label{ass:finite-lda}
The observations \((Y_i,x_{i1},\ldots,x_{iN})\), \(i=1,\ldots,n\), are independent and identically distributed according to the model in Section~\ref{sec:model}. The integers \(d\), \(k\), and \(N\ge3\) are fixed. The topic matrix \(O\in\R^{d\times k}\) has full column rank, each column lies in the vocabulary simplex, \(\alpha\in(0,\infty)^k\), and \(\E Y^2<\infty\). The response satisfies \eqref{eq:supervised-linear-model} and the low-order orthogonality conditions \eqref{eq:response-token-orthogonality}; a sufficient condition is \eqref{eq:strong-response-condition}.
\end{assumption}

\begin{assumption}[Concentration parameter and finite probes]\label{ass:alpha-probes}
The number of topics satisfies \(k\ge3\). The compact interval \(\Theta=[\underline\tau,\overline\tau]\subset(0,\infty)\) contains \(\alpha_0\) in its interior. For the probe vectors \(s_1,\ldots,s_L\), \(L\ge2\), and the tested pairs \(\mathcal I\), at least one pair \((\ell,q)\in\mathcal I\) satisfies that
\[
    O^\top P(\mu)s_\ell
    \quad\hbox{and}\quad
    O^\top P(\mu)s_q
\]
are non-collinear, where \(P(u)=I_d-uu^\top/\|u\|^2\).
\end{assumption}

\begin{assumption}[Topic ordering]\label{ass:ordering}
The ordering direction \(r\in\R^d\) satisfies that the entries of \(O^\top P(\mu)r\) are distinct and nonzero.
\end{assumption}

\begin{assumption}[Rank threshold]\label{ass:rank-threshold}
If \(k\) is estimated by \eqref{eq:est-khat}, the threshold satisfies
\[
    a_n\to0,
    \qquad
    \sqrt n\,a_n\to\infty .
\]
\end{assumption}

\begin{assumption}[Split compression and observed controls]\label{ass:split-compressed-controls}
For the compressed control-adjusted estimator, split the sample into an auxiliary
part \(I_1\) and an estimation part \(I_2\), with \(|I_2|=n_2\to\infty\) and
\(n_2/n\) bounded away from zero. The compression dimension \(m\), the number of
controls \(p\), and the selected working rank are fixed for the second-split
analysis. The compression matrix \(\hat R_1\), the working rank, the finite probe
set, and the ordering direction are measurable with respect to \(I_1\), or are
chosen by external randomization independent of \(I_2\). Conditional on these
first-split choices, the selected working rank equals the population rank \(k\),
\(\sigma_{\min}(\hat R_1^\top O)\) is bounded away from zero, the compressed
finite-probe non-collinearity condition holds, and the selected nonzero-cluster
ordering eigenvalues are real, simple, mutually separated, and separated from
zero. For the observed-control adjustment, the compressed scale-recovery
denominators are bounded away from zero and the corresponding matrix
\(\Gamma_q(\hat R_1)\) is nonsingular. The enlarged second-split moment vector
\(Z^c_{\hat R_1,i}\), containing the compressed word, response-weighted word,
control, control-response, control-word, and scale-recovery moments entering the
estimator, has finite conditional second moment. With unequal document lengths,
\(Z^c_{\hat R_1,i}\) is formed using the length-normalized moments in
Appendix~\ref{app:variable-lengths}.
\end{assumption}

\begin{remark}[Random probes and ordering directions]\label{rem:random-probes}
If the probes \(s_1,s_2\) are drawn independently from distributions that are absolutely continuous on \(\R^d\), Assumption~\ref{ass:alpha-probes} holds with probability one. The map \(s\mapsto O^\top P(\mu)s\) has image \(\alpha^\perp\). Since \(k\ge3\), this image has dimension at least two, so two independent absolutely continuous projected draws are non-collinear with probability one. The same finite-union-of-hyperplanes argument gives Assumption~\ref{ass:ordering} with probability one for an absolutely continuous draw of \(r\).
\end{remark}

\subsection{Dirichlet and observable word moments}\label{app:dirichlet-word-moments}

Let \(h\sim\operatorname{Dirichlet}(\alpha)\), write \(\alpha_0=\ones^\top\alpha\), \(D=\diag(\alpha)\), and let \((a)_m=a(a+1)\cdots(a+m-1)\) be the rising factorial. The Dirichlet factorial-moment formula is
\begin{equation}\label{eq:app-dirichlet-factorial}
    \E\left\{\prod_{j=1}^k h_j^{m_j}\right\}
    =\frac{\prod_{j=1}^k(\alpha_j)_{m_j}}{(\alpha_0)_{m_+}},
    \qquad
    m_+=\sum_{j=1}^k m_j .
\end{equation}
Consequently,
\begin{equation}\label{eq:app-dirichlet-second}
    \E h=\frac{\alpha}{\alpha_0},
    \qquad
    \E(hh^\top)=\frac{D+\alpha\alpha^\top}{\alpha_0(\alpha_0+1)} .
\end{equation}
For \(\gamma\in\R^k\), the contracted third moment is
\begin{equation}\label{eq:app-dirichlet-third}
\begin{aligned}
\E\{hh^\top(\gamma^\top h)\}
=\frac{1}{\alpha_0(\alpha_0+1)(\alpha_0+2)}
\Big
(& (\alpha^\top\gamma)\alpha\alpha^\top
  +\alpha\alpha^\top\diag(\gamma)
  +\diag(\gamma)\alpha\alpha^\top  \\
& +(\alpha^\top\gamma)D
  +2\diag(\alpha\circ\gamma)
\Big).
\end{aligned}
\end{equation}
For example, the \((i,j)\) entry of the right side is \(\alpha_i\alpha_j\{\alpha^\top\gamma+\gamma_i+\gamma_j\}/(\alpha_0)_3\) when \(i\ne j\), and \(\alpha_i(\alpha_i+1)\{\alpha^\top\gamma+2\gamma_i\}/(\alpha_0)_3\) when \(i=j\), which is exactly \eqref{eq:app-dirichlet-factorial}.

Under the word model, conditional independence of distinct tokens gives
\[
    \E(x_a\mid h)=Oh,
    \qquad
    \E(x_ax_b^\top\mid h)=(Oh)(Oh)^\top,
    \quad a\ne b .
\]
Thus
\begin{equation}\label{eq:app-M2-mu-topic}
    \mu=O\frac{\alpha}{\alpha_0},
    \qquad
    M_2=O\frac{D+\alpha\alpha^\top}{\alpha_0(\alpha_0+1)}O^\top .
\end{equation}
For \(\eta\in\R^d\), write \(w=O^\top\eta\), \(r=\alpha^\top w\), and \(s=\alpha\circ w=Dw\). Equation~\eqref{eq:app-dirichlet-third} implies
\begin{equation}\label{eq:app-T-topic}
    T(\eta)
    =\frac1{C_3}
      O\{r\alpha\alpha^\top+\alpha s^\top+s\alpha^\top+rD+2\diag(s)\}O^\top,
    \qquad
    C_3=\alpha_0(\alpha_0+1)(\alpha_0+2).
\end{equation}
The supervised moments obey the bridge identities in Section~\ref{sec:model}. Indeed, under \eqref{eq:supervised-linear-model} and the low-order orthogonality conditions \eqref{eq:response-token-orthogonality},
\begin{align}
    m_y&=\beta^\top\alpha/\alpha_0,\label{eq:app-supervised-my}\\
    v_y&=O\E(hh^\top)\beta,\label{eq:app-supervised-vy}\\
    T^y&=O\E\{hh^\top(\beta^\top h)\}O^\top .\label{eq:app-supervised-Ty}
\end{align}
Therefore, whenever \(O^\top\eta_\beta=\beta\),
\begin{equation}\label{eq:app-bridge}
    m_y=\langle \eta_\beta,\mu\rangle,
    \qquad
    v_y=M_2\eta_\beta,
    \qquad
    T^y=T(\eta_\beta).
\end{equation}

\subsection{Corrected moment factorization}\label{app:corrected-factorization}

This subsection proves Lemma~\ref{lem:id-corrected-factorization}. Put \(C_2=\alpha_0(\alpha_0+1)\). Substituting \eqref{eq:app-M2-mu-topic} into \eqref{eq:id-Btau} gives
\begin{equation}\label{eq:app-SB}
    B_\tau=OS_B(\tau)O^\top,
    \qquad
    S_B(\tau)=\frac{D}{C_2}
    +\frac{\alpha_0-\tau}{\alpha_0^2(\alpha_0+1)(\tau+1)}\alpha\alpha^\top .
\end{equation}
The inverse of \(S_B(\tau)\) is
\begin{equation}\label{eq:app-SB-inverse}
    S_B(\tau)^{-1}=C_2D^{-1}-(\alpha_0-\tau)\ones\ones^\top .
\end{equation}
This follows from the Sherman--Morrison formula and the identities \(D^{-1}\alpha=\ones\) and \(\alpha^\top D^{-1}\alpha=\alpha_0\). Positive definiteness holds for every \(\tau>0\). If \(\tau\le\alpha_0\), it is immediate from \eqref{eq:app-SB}. If \(\tau>\alpha_0\), Cauchy's inequality gives \((a^\top\alpha)^2\le\alpha_0 a^\top Da\), hence
\[
    a^\top S_B(\tau)a
    \ge
    \left\{\frac{1}{C_2}
    +\frac{\alpha_0(\alpha_0-\tau)}{\alpha_0^2(\alpha_0+1)(\tau+1)}\right\}a^\top Da
    =\frac{a^\top Da}{\alpha_0(\tau+1)}>0 .
\]

The remaining terms in \(A_\tau(\eta)\) can be written in topic coordinates as
\begin{align}
    M_2\eta\mu^\top
    &=O\frac{(s+r\alpha)\alpha^\top}{C_2\alpha_0}O^\top,\label{eq:app-third-corr-1}\\
    \mu\eta^\top M_2
    &=O\frac{\alpha(s+r\alpha)^\top}{C_2\alpha_0}O^\top,\label{eq:app-third-corr-2}\\
    \langle\eta,\mu\rangle M_2
    &=O\frac{r(D+\alpha\alpha^\top)}{C_2\alpha_0}O^\top,\label{eq:app-third-corr-3}\\
    \langle\eta,\mu\rangle\mu\mu^\top
    &=O\frac{r\alpha\alpha^\top}{\alpha_0^3}O^\top .\label{eq:app-third-corr-4}
\end{align}
Combining \eqref{eq:app-T-topic} and \eqref{eq:app-third-corr-1}--\eqref{eq:app-third-corr-4},
\begin{equation}\label{eq:app-SA-general-expanded}
    A_\tau(\eta)=OS_A(\tau;w)O^\top,
\end{equation}
where
\begin{equation}\label{eq:app-SA-general}
\begin{aligned}
S_A(\tau;w)
={}&\frac{1}{C_3}
  \{r\alpha\alpha^\top+\alpha s^\top+s\alpha^\top+rD+2\diag(s)\}  \\
&-\frac{\tau}{(\tau+2)\alpha_0^2(\alpha_0+1)}
  \{\alpha s^\top+s\alpha^\top+3r\alpha\alpha^\top+rD\} \\
&+\frac{2\tau^2r}{(\tau+2)(\tau+1)\alpha_0^3}\alpha\alpha^\top .
\end{aligned}
\end{equation}
The cancellation at the true concentration is the essential point. At \(\tau=\alpha_0\), the coefficients of \(\alpha s^\top+s\alpha^\top\), \(rD\), and \(r\alpha\alpha^\top\) in \eqref{eq:app-SA-general} are respectively
\[
    \frac{1}{C_3}-\frac{1}{\alpha_0(\alpha_0+1)(\alpha_0+2)},
\]
\[
    \frac{1}{C_3}-\frac{1}{\alpha_0(\alpha_0+1)(\alpha_0+2)},
\]
 and
\[
    \frac{1}{C_3}
    -\frac{3}{\alpha_0(\alpha_0+1)(\alpha_0+2)}
    +\frac{2}{\alpha_0(\alpha_0+1)(\alpha_0+2)},
\]
all of which are zero. Only the diagonal term remains:
\begin{equation}\label{eq:app-SA-true-diagonal}
    S_A(\alpha_0;w)=\frac{2}{C_3}\diag(\alpha\circ w).
\end{equation}

It remains to pass from topic coordinates to the observed operator. Since \(O\) has full column rank and \(S_B(\tau)\succ0\),
\begin{equation}\label{eq:app-B-pinv}
    (OS_B(\tau)O^\top)^+=O^{+\top}S_B(\tau)^{-1}O^+ .
\end{equation}
The Moore--Penrose equations follow directly from \(O^+O=I_k\) and from the fact that \(OO^+\) is the orthogonal projector onto the column space of \(O\). Therefore
\begin{equation}\label{eq:app-H-topic-conjugacy}
    H_\tau(\eta)=A_\tau(\eta)B_\tau^+
    =O\{S_A(\tau;w)S_B(\tau)^{-1}\}O^+ .
\end{equation}
Using \eqref{eq:app-SB-inverse} and \eqref{eq:app-SA-true-diagonal} at \(\tau=\alpha_0\),
\begin{equation}\label{eq:app-H-true-diagonal}
    H_{\alpha_0}(\eta)
    =O\left\{\frac{2}{\alpha_0+2}\diag(O^\top\eta)\right\}O^+ .
\end{equation}
This proves Lemma~\ref{lem:id-corrected-factorization}.

\subsection{Population identification proofs}\label{app:population-identification}

\begin{proof}[Proof of Theorem~\ref{thm:id-known-alpha-topic}]
Equation~\eqref{eq:app-H-true-diagonal} implies
\[
    H_{\alpha_0}(\eta)O_{:j}
    =\frac{2}{\alpha_0+2}(O^\top\eta)_jO_{:j},
    \qquad j=1,\ldots,k .
\]
Thus the topic columns are right eigenvectors. If the coordinates of \(O^\top\eta\) are distinct and nonzero, these \(k\) nonzero eigenvalues are simple and are separated from the zero eigenvalue on the null space of \(H_{\alpha_0}(\eta)\). The exceptional set is a finite union of proper hyperplanes in \(\R^d\): the hyperplanes \((O_{:j}-O_{:\ell})^\top\eta=0\), \(j\ne\ell\), and \(O_{:j}^\top\eta=0\). An absolutely continuous draw of \(\eta\) avoids this set with probability one. The simplex constraint \(\ones^\top O_{:j}=1\) fixes the scale of each eigenvector.
\end{proof}

\begin{proof}[Proof of Theorem~\ref{thm:id-beta-supervised}]
Choose \(\eta_\beta\) with \(O^\top\eta_\beta=\beta\). The bridge identities \eqref{eq:app-bridge} imply that \(A_\tau^y=A_\tau(\eta_\beta)\) for every \(\tau\), after replacing \(\langle\eta_\beta,\mu\rangle\), \(M_2\eta_\beta\), and \(T(\eta_\beta)\) by \(m_y\), \(v_y\), and \(T^y\). At \(\tau=\alpha_0\), \eqref{eq:app-SA-true-diagonal} gives
\[
    A_{\alpha_0}^y
    =O\frac{2}{C_3}\diag(\alpha\circ\beta)O^\top .
\]
Multiplication by \(B_{\alpha_0}^+=O^{+\top}C_2D^{-1}O^+\) yields
\[
    H_{\alpha_0}^y
    =O\left\{\frac{2}{\alpha_0+2}\diag(\beta)\right\}O^+ .
\]
Therefore \(O^+H_{\alpha_0}^yO=2\diag(\beta)/(\alpha_0+2)\), and \eqref{eq:id-beta-projection} follows.
\end{proof}

\begin{proof}[Proof of Theorem~\ref{thm:id-alpha-commutativity}]
At \(\tau=\alpha_0\), \eqref{eq:app-H-true-diagonal} gives simultaneous diagonalization in the topic basis, so the family \(\{H_{\alpha_0}(\eta):\eta\in\R^d\}\) is pairwise commuting.

Now fix \(\tau\ne\alpha_0\) and let \(w\in\alpha^\perp\). Then \(r=\alpha^\top w=0\), and \eqref{eq:app-SA-general} reduces to
\begin{equation}\label{eq:app-SA-alpha-perp}
    S_A(\tau;w)
    =\frac{2}{C_3}\diag(\alpha\circ w)
     +\kappa_\tau\{\alpha(\alpha\circ w)^\top+(\alpha\circ w)\alpha^\top\},
\end{equation}
where
\begin{equation}\label{eq:app-kappa-tau}
    \kappa_\tau=
\frac{2(\alpha_0-\tau)}{\alpha_0^2(\alpha_0+1)(\alpha_0+2)(\tau+2)} .
\end{equation}
Set \(s=\alpha\circ w\). Since \(s^\top\ones=\alpha^\top w=0\), \(s^\top D^{-1}=w^\top\), and \(D^{-1}\alpha=\ones\), direct multiplication with \eqref{eq:app-SB-inverse} gives
\begin{equation}\label{eq:app-Ttau-alpha-perp}
    T_\tau(w):=S_A(\tau;w)S_B(\tau)^{-1}
    =\lambda\diag(w)+\beta_\tau\alpha w^\top+\rho_\tau(\alpha\circ w)\ones^\top,
\end{equation}
where
\begin{equation}\label{eq:app-lambda-beta-rho}
    \lambda=\frac{2}{\alpha_0+2},
    \qquad
    \beta_\tau=\frac{2(\alpha_0-\tau)}{\alpha_0(\alpha_0+2)(\tau+2)},
    \qquad
    \rho_\tau=-\frac{2(\alpha_0-\tau)}{\alpha_0(\alpha_0+1)(\alpha_0+2)(\tau+2)} .
\end{equation}
For \(w_1,w_2\in\alpha^\perp\), define \(s_j=\alpha\circ w_j\), \(\Delta_j=\diag(w_j)\), \(R_j=\alpha w_j^\top\), and \(Q_j=s_j\ones^\top\). Then \(T_\tau(w_j)=\lambda\Delta_j+\beta_\tau R_j+\rho_\tau Q_j\). The like-term commutators vanish:
\[
    [\Delta_1,\Delta_2]=[R_1,R_2]=[Q_1,Q_2]=0,
\]
because \(w_j^\top\alpha=0\) and \(\ones^\top s_j=0\). The cross terms are
\begin{align}
    [\Delta_1,R_2]+[R_1,\Delta_2]
    &=s_1w_2^\top-s_2w_1^\top,\label{eq:app-cross-DR}\\
    [\Delta_1,Q_2]+[Q_1,\Delta_2]
    &=s_1w_2^\top-s_2w_1^\top,\label{eq:app-cross-DQ}\\
    [R_1,Q_2]+[Q_1,R_2]
    &=\alpha_0(s_1w_2^\top-s_2w_1^\top).\label{eq:app-cross-RQ}
\end{align}
Combining \eqref{eq:app-Ttau-alpha-perp}--\eqref{eq:app-cross-RQ},
\begin{equation}\label{eq:app-commutator-formula}
    [T_\tau(w_1),T_\tau(w_2)]
    =c(\tau,\alpha_0)
    \{(\alpha\circ w_1)w_2^\top-(\alpha\circ w_2)w_1^\top\},
\end{equation}
with
\begin{equation}\label{eq:app-c-tau-alpha}
    c(\tau,\alpha_0)
    =\frac{4(\alpha_0-\tau)(\alpha_0\tau+\alpha_0+\tau)}
    {\alpha_0(\alpha_0+1)(\alpha_0+2)^2(\tau+2)^2} .
\end{equation}
The scalar \(c(\tau,\alpha_0)\) is nonzero for every \(\tau>0\) with \(\tau\ne\alpha_0\). Moreover,
\[
    (\alpha\circ w_1)w_2^\top-(\alpha\circ w_2)w_1^\top
    =D(w_1w_2^\top-w_2w_1^\top),
\]
which is nonzero if and only if \(w_1\) and \(w_2\) are non-collinear. This proves the latent-space noncommutativity claim.

To lift the result to observed coordinates, note that \(\eta\in\mu^\perp\) implies \(O^\top\eta\in\alpha^\perp\), since \(\mu=O\alpha/\alpha_0\). Conversely, every \(w\in\alpha^\perp\) is of the form \(w=O^\top\eta\) for some \(\eta\in\mu^\perp\): take \(\eta=O(O^\top O)^{-1}w\). Finally, \eqref{eq:app-H-topic-conjugacy} gives
\[
    [H_\tau(\eta_1),H_\tau(\eta_2)]
    =O[T_\tau(w_1),T_\tau(w_2)]O^+ .
\]
Premultiplication by \(O^+\) and postmultiplication by \(O\) show that the observed-space commutator is zero only if the latent-space commutator is zero. This proves the theorem.
\end{proof}

\begin{lemma}[Finite-probe identification and local rank condition]\label{lem:finite-probe-local}
Under Assumption~\ref{ass:alpha-probes}, the population finite-probe criterion \eqref{eq:id-pop-criterion} satisfies
\[
    Q_{\mathcal I}(\tau)=0
    \quad\Longleftrightarrow\quad
    \tau=\alpha_0 .
\]
Moreover, if
\[
    g_{\mathcal I}(\tau)
    :=\stackop_{(\ell,q)\in\mathcal I}
      \vecop\{[H_\tau(v_\ell),H_\tau(v_q)]\},
    \qquad
    v_\ell=P(\mu)s_\ell,
\]
then \(G_\tau:=\partial_\tau g_{\mathcal I}(\alpha_0)\ne0\).
\end{lemma}

\begin{proof}
The equivalence follows from Theorem~\ref{thm:id-alpha-commutativity} and from the existence of one tested pair whose latent projections are non-collinear. For the derivative, differentiate \eqref{eq:app-commutator-formula} at \(\tau=\alpha_0\). Since
\begin{equation}\label{eq:app-cprime-alpha}
    \left.\partial_\tau c(\tau,\alpha_0)\right|_{\tau=\alpha_0}
    =-\frac{4}{(\alpha_0+1)(\alpha_0+2)^3}\ne0,
\end{equation}
the derivative of the tested commutator for a non-collinear pair is nonzero. Hence the stacked derivative \(G_\tau\) is nonzero.
\end{proof}

\subsection{Empirical moment estimators and primitive CLT}\label{app:empirical-moments-clt}

For one document define
\[
    \hat\mu_i=\frac1N\sum_{a=1}^N x_{ia},
    \qquad
    \hat M_{2,i}=\frac{1}{N(N-1)}\sum_{a\ne b}x_{ia}x_{ib}^\top,
\]
and the ordered-distinct third tensor
\[
    \hat{\mathcal T}_i
    =\frac{1}{N(N-1)(N-2)}
      \sum_{a,b,c\;\mathrm{distinct}}x_{ia}\otimes x_{ib}\otimes x_{ic} .
\]
Then \(\hat T_i(v)=\hat{\mathcal T}_i(I,I,v)\). Conditional on \(h_i\), exchangeability and conditional independence give
\begin{align*}
    \E(\hat\mu_i\mid h_i)&=Oh_i,\\
    \E(\hat M_{2,i}\mid h_i)&=(Oh_i)(Oh_i)^\top,\\
    \E\{\hat T_i(v)\mid h_i\}&=(Oh_i)(Oh_i)^\top\{(O^\top v)^\top h_i\}.
\end{align*}
Unconditioning gives \(\E\hat\mu_i=\mu\), \(\E\hat M_{2,i}=M_2\), and \(\E\hat T_i(v)=T(v)\). Likewise, using the response model and the positionwise orthogonality conditions in \eqref{eq:response-token-orthogonality},
\[
    \E(Y_i)=m_y,
    \qquad
    \E(Y_i\hat\mu_i)=v_y,
    \qquad
    \E(Y_i\hat M_{2,i})=T^y .
\]
The use of ordered distinct token positions removes the single-token multinomial diagonal contribution from \(\hat M_{2,i}\).

Fix any linear vectorization of matrices and third-order tensors and define the finite-dimensional document-level vector
\begin{equation}\label{eq:app-Zi-vector}
    Z_i
    :=\left(
        \hat\mu_i,
        \hat M_{2,i},
        \hat{\mathcal T}_i,
        Y_i,
        Y_i\hat\mu_i,
        Y_i\hat M_{2,i}
      \right) .
\end{equation}
Let \(Z_0=\E Z_i\), \(\bar Z=n^{-1}\sum_iZ_i\), and \(\Omega=\Var(Z_i)\). The word components are bounded because tokens are one-hot and \(N\) is fixed. Hence \(\E\|Z_i\|^2<\infty\) whenever \(\E Y^2<\infty\). Since documents are iid,
\begin{equation}\label{eq:app-Z-clt}
    \sqrt n(\bar Z-Z_0)\rightsquigarrow N(0,\Omega).
\end{equation}
The empirical covariance
\begin{equation}\label{eq:app-Omega-hat}
    \hat\Omega=\frac1n\sum_{i=1}^n(Z_i-\bar Z)(Z_i-\bar Z)^\top
\end{equation}
satisfies \(\hat\Omega\xrightarrow{p}\Omega\). This is an iid-document CLT. The within-document averages are U-statistic-type unbiased moment estimators, but no asymptotics are taken in \(N\).

\subsection{Finite-dimensional plug-in maps}\label{app:finite-dimensional-maps}

This subsection makes the estimator a deterministic finite-dimensional map of empirical moments. The map is defined on the symmetric moment subspace
\[
\mathcal Z_{\rm sym}=\{(u,M,{\cal T},m,v,S):
M=M^\top,\; S=S^\top,\; {\cal T}(a)={\cal T}(a)^\top\text{ for every }a\in\R^d\}.
\]
A generic point in this relative moment space is written as
\[
    z=(u,M,{\cal T},m,v,S),
\]
where \(u\in\R^d\), \(M\in\R^{d\times d}\), \({\cal T}\) is a third-order tensor represented through its contractions \({\cal T}(a)={\cal T}(I,I,a)\), \(m\in\R\), \(v\in\R^d\), and \(S\in\R^{d\times d}\), subject to the symmetry restrictions above. The population and empirical moments lie in \(\mathcal Z_{\rm sym}\), and all neighborhoods and derivatives below are relative to this finite-dimensional subspace. At the population value,
\[
    z_0=(\mu,M_2,{\cal T},m_y,v_y,T^y),
    \qquad
    {\cal T}(a)=T(a).
\]
At the sample value, \(z=\bar Z\). For \(u\ne0\), set \(v_\ell(u)=P(u)s_\ell\). For candidate \(\tau\), define
\begin{align*}
    B_\tau(z)&=M-\frac{\tau}{\tau+1}uu^\top,\\
    A_\tau(a;z)
    &={\cal T}(a)
      -\frac{\tau}{\tau+2}\{Mau^\top+ua^\top M+\langle a,u\rangle M\}
      +\frac{2\tau^2}{(\tau+1)(\tau+2)}\langle a,u\rangle uu^\top,\\
    A_\tau^y(z)
    &=S-\frac{\tau}{\tau+2}\{vu^\top+uv^\top+mM\}
      +\frac{2\tau^2}{(\tau+1)(\tau+2)}muu^\top .
\end{align*}
When \(B_\tau(z)\) has a positive \(k\)-th eigenvalue separated from the remaining eigenvalues, let \(B_\tau(z)^{(k)+}\) denote the rank-\(k\) truncated pseudoinverse. Then
\[
    H_\tau(a;z)=A_\tau(a;z)B_\tau(z)^{(k)+},
    \qquad
    H_\tau^y(z)=A_\tau^y(z)B_\tau(z)^{(k)+}.
\]
The finite-probe commutator map is
\begin{equation}\label{eq:app-g-finite-map}
    g(\tau,z)
    :=\stackop_{(\ell,q)\in\mathcal I}
      \vecop\left(
      [H_\tau\{v_\ell(u);z\},H_\tau\{v_q(u);z\}]
      \right),
\end{equation}
and the population criterion is \(Q_{\mathcal I}(\tau)=\|g(\tau,z_0)\|^2\). On \(\{\hat k=k\}\), the empirical criterion in \eqref{eq:est-Q} is \(\|g(\tau,\bar Z)\|^2\), apart from the finite-sample convention that it is set to \(+\infty\) when the estimated \(k\)-th corrected eigenvalue is nonpositive. That convention is asymptotically irrelevant because the positive spectral cluster of \(B_\tau\) is uniformly separated from zero on \(\Theta\).

The topic-ordering map is defined by
\[
    h^o(z,\tau):=H_\tau\{P(u)r;z\}.
\]
On a neighborhood of \((z_0,\alpha_0)\) where the \(k\) nonzero ordering eigenvalues are real, simple, mutually separated, and separated from the zero cluster, the selected nonzero right eigenvectors are labeled by decreasing real eigenvalue within this selected cluster. Their simplex-normalized columns define a matrix-valued map \(O_r(z,\tau)\). Finally,
\begin{equation}\label{eq:app-beta-full-map}
    b_r(z,\tau)
    :=\frac{\tau+2}{2}\diag\left[
       O_r(z,\tau)^+ H_\tau^y(z)O_r(z,\tau)
    \right].
\end{equation}
The estimator \(\hat\beta(r)\) is \(b_r(\bar Z,\hat\alpha_0)\).

\subsection{Smoothness lemmas}\label{app:smoothness}

The following facts replace a coordinate-by-coordinate derivative catalogue. They are standard finite-dimensional perturbation results; they are stated here to make clear where differentiability enters the delta-method proof.

\begin{lemma}[Smooth truncated pseudoinverse]\label{lem:app-truncated-pinv}
Let \(\mathcal U\) be an open subset of a finite-dimensional Euclidean space, and let
\(B:\mathcal U\to\R^{d\times d}\) be a \(C^r\) map, \(r\ge1\), whose values are real
symmetric matrices. Suppose that, for every \(\theta\in\mathcal U\), the top \(k\)
eigenvalues of \(B(\theta)\) form a positive spectral cluster separated from the
remaining spectrum: \(\lambda_k\{B(\theta)\}>0\) and, if \(k<d\),
\(\lambda_k\{B(\theta)\}>\lambda_{k+1}\{B(\theta)\}\). Then the rank-\(k\)
truncated pseudoinverse \(B(\theta)^{(k)+}\) is \(C^r\) on \(\mathcal U\). In
particular, if \(B\) is \(C^2\) in a scalar coordinate and \(C^1\) jointly in all
coordinates, then \(B(\theta)^{(k)+}\) has the same regularity.
\end{lemma}

\begin{proof}
The claim is local. Fix \(\theta_0\in\mathcal U\), and choose a positively oriented
contour \(\Gamma\) in the complex plane that encloses the positive top-\(k\) cluster
of \(B(\theta_0)\) and excludes the remaining spectrum and zero. By spectral
continuity, the same contour separates the cluster for all \(\theta\) in a
neighborhood of \(\theta_0\). On this neighborhood,
\[
    B(\theta)^{(k)+}
    =\frac{1}{2\pi i}\int_\Gamma
      \zeta^{-1}\{\zeta I-B(\theta)\}^{-1}\,d\zeta .
\]
The resolvent is \(C^r\) in \(\theta\) because inversion is smooth on nonsingular
matrices and \(\zeta I-B(\theta)\) is nonsingular uniformly for \(\zeta\in\Gamma\).
Differentiation under the finite contour integral gives the asserted smoothness.
When the lemma is invoked on a compact set \(\mathcal U_0\subset\mathcal U\) on which the spectral-gap hypothesis holds uniformly, finitely many such local contours cover \(\mathcal U_0\), and the same conclusion (with derivative bounds uniform on \(\mathcal U_0\)) follows on each piece of the cover.
\end{proof}

\begin{lemma}[Smooth normalized right eigenvectors]\label{lem:app-eigenvector-smooth}
Let \(H(t)\) be a \(C^1\) curve of real matrices. Suppose \(H(0)o_j=\lambda_jo_j\), where \(\lambda_j\) is simple and \(\ones^\top o_j=1\). Then the normalized right eigenvector is \(C^1\) near zero.
\end{lemma}

\begin{proof}
The hypothesis \(\ones^\top o_j=1\) implies in particular \(\ones^\top o_j\ne0\), which is what makes the affine normalization transverse to the eigenline. Apply the implicit function theorem to
\[
    F(o,\lambda;t):=\bigl(H(t)o-\lambda o,\;\ones^\top o-1\bigr)=0 .
\]
The Jacobian with respect to \((o,\lambda)\) at \((o_j,\lambda_j,0)\) is
\[
    J=\begin{pmatrix} H(0)-\lambda_j I & -o_j\\ \ones^\top & 0\end{pmatrix}.
\]
Take \((v,\mu)\in\ker J\), so that \((H(0)-\lambda_j I)v=\mu\,o_j\) and \(\ones^\top v=0\). If \(\mu=0\), then \(v\) lies in the right-eigenspace of \(\lambda_j\), which is one-dimensional because \(\lambda_j\) is simple; combined with \(\ones^\top v=0\) and \(\ones^\top o_j=1\), this forces \(v=0\). If \(\mu\ne0\), solvability of \((H(0)-\lambda_j I)v=\mu\,o_j\) requires \(o_j\) to lie in the range of \(H(0)-\lambda_j I\). For a simple eigenvalue, that range is \(\{u:\ell_j^\top u=0\}\), where \(\ell_j\) is the left eigenvector dual to \(o_j\) with \(\ell_j^\top o_j=1\); hence \(o_j\) is not in the range, contradicting \(\mu\ne0\). Therefore \(\ker J=\{0\}\), \(J\) is nonsingular, and the implicit function theorem yields a \(C^1\) map \(t\mapsto(o_j(t),\lambda_j(t))\) near zero.
\end{proof}

\begin{lemma}[Smooth estimator maps]\label{lem:app-smooth-estimator-maps}
Under Assumptions~\ref{ass:finite-lda} and \ref{ass:alpha-probes}, the following
commutator maps are smooth enough for the concentration expansion: there is a
neighborhood \({\cal K}\) of \(z_0\) such that, for each fixed probe \(s_\ell\),
\[
    (\tau,z)\mapsto H_\tau(P(u)s_\ell;z),
    \qquad
    (\tau,z)\mapsto H_\tau^y(z)
\]
are twice continuously differentiable in \(\tau\) and continuously differentiable in \(z\)
on \(\Theta\times{\cal K}\). Consequently the finite-probe commutator map \(g(\tau,z)\)
in \eqref{eq:app-g-finite-map} has the same smoothness on \(\Theta\times{\cal K}\), and the
first-order-condition map
\[
        S(\tau,z):=\{\partial_\tau g(\tau,z)\}^\top g(\tau,z)
\]
is continuously differentiable in a neighborhood of \((\alpha_0,z_0)\).

If Assumption~\ref{ass:ordering} also holds, then there is a neighborhood \({\cal N}\) of
\((\alpha_0,z_0)\) such that the local topic map \((\tau,z)\mapsto O_r(z,\tau)\) and the
coefficient map \((\tau,z)\mapsto b_r(z,\tau)\) are continuously differentiable on \({\cal N}\).
\end{lemma}

\begin{proof}
The maps \((z,\tau)\mapsto B_\tau(z)\), \(A_\tau(a;z)\), \(A_\tau^y(z)\), and \(P(u)s_\ell\)
are rational or polynomial maps of finite-dimensional arguments and are smooth when \(u\ne0\)
and \(\tau>0\). By \eqref{eq:app-SB}, \(B_\tau(z_0)\) has exactly \(k\) positive eigenvalues
for every \(\tau\in\Theta\), and the smallest positive eigenvalue is bounded away from zero
uniformly on \(\Theta\) by Lemma~\ref{lem:app-uniform-spectral-gap}. After shrinking
\({\cal K}\), the same spectral separation holds for \(B_\tau(z)\) uniformly over
\(\Theta\times{\cal K}\). Lemma~\ref{lem:app-truncated-pinv} then gives the required
\(C^2\)-in-\(\tau\) and \(C^1\)-in-\(z\) smoothness of the truncated inverse on this set, hence of \(H_\tau\), \(H_\tau^y\), and the
finite collection of commutators defining \(g\). The displayed smoothness of \(S\) follows by
composition. If Assumption~\ref{ass:ordering} holds, then at \((z_0,\alpha_0)\) the ordering
operator has simple nonzero eigenvalues that are separated from the zero cluster, so
Lemma~\ref{lem:app-eigenvector-smooth} gives differentiability of \(O_r\) after
restricting to a local neighborhood \({\cal N}\). The
Moore--Penrose inverse is smooth on the full-column-rank stratum, so \(b_r\) is differentiable
by composition.
\end{proof}

\begin{lemma}[Topic perturbations do not enter the diagonal coefficient map]\label{lem:app-beta-cancellation}
Let
\[
    b(\tau,O,H)=\frac{\tau+2}{2}\diag(O^+HO),
\]
where \(O\) has full column rank. At the true supervised operator \(H^y=O\Lambda_yO^+\), with \(\Lambda_y=2\diag(\beta)/(\alpha_0+2)\), the first-order contribution of a perturbation of \(O\) to the diagonal of \(O^+H^yO\) is zero. Consequently, for perturbations \((\dot\tau,\dot O,\dot H)\),
\begin{equation}\label{eq:app-db-truth}
    Db[\dot\tau,\dot O,\dot H]
    =\frac{\dot\tau}{\alpha_0+2}\beta
     +\frac{\alpha_0+2}{2}\diag(O^+\dot H O)
\end{equation}
at the truth.
\end{lemma}

\begin{proof}
Let \(R=O^+\dot O\). Differentiating \(O^+O=I_k\) gives \(D(O^+)[\dot O]\,O=-R\). At the population value, \(H^y=O\Lambda_yO^+\) with \(O^+O=I_k\), so \(H^yO=O\Lambda_y\) and \(O^+H^y=\Lambda_yO^+\). Hence the first-order contribution of \(\dot O\) to \(O^+H^yO\) is
\[
    D(O^+)[\dot O]\,H^yO + O^+H^y\,\dot O
    = D(O^+)[\dot O]\,O\Lambda_y + \Lambda_y\,O^+\dot O
    =-R\Lambda_y+\Lambda_y R
    =[\Lambda_y,R].
\]
This is a commutator with a diagonal matrix, so its diagonal is zero. Differentiating the multiplicative factor \((\tau+2)/2\) gives the first term in \eqref{eq:app-db-truth}, and differentiating \(H\) gives the second.
\end{proof}

\subsection{A smooth minimum-distance expansion}\label{app:minimum-distance-lemma}

The concentration estimator is a one-dimensional smooth minimum-distance estimator.

\begin{lemma}[One-dimensional minimum-distance expansion]\label{lem:app-md-expansion}
Let \(g(\tau,z)\in\R^q\), and define
\[
        S(\tau,z):=\{\partial_\tau g(\tau,z)\}^\top g(\tau,z).
\]
Suppose \(S\) is continuously differentiable in a neighborhood of \((\alpha_0,z_0)\),
\(g(\alpha_0,z_0)=0\), and \(G_\tau:=\partial_\tau g(\alpha_0,z_0)\ne0\). Let
\(G_z:=D_zg(\alpha_0,z_0)\). If \(\hat z-z_0=O_p(n^{-1/2})\), \(\hat\alpha\to_p\alpha_0\), and,
with probability tending to one, \(\hat\alpha\) is an interior solution of
\[
        S(\hat\alpha,\hat z)=0,
\]
then
\begin{equation}\label{eq:app-md-expansion}
    \sqrt n(\hat\alpha-\alpha_0)
    =-(G_\tau^\top G_\tau)^{-1}G_\tau^\top G_z\sqrt n(\hat z-z_0)+o_p(1).
\end{equation}
\end{lemma}

\begin{proof}
At the population value, \(g(\alpha_0,z_0)=0\), so \(S(\alpha_0,z_0)=0\). Differentiating \(S=(\partial_\tau g)^\top g\),
\[
    \partial_\tau S=\bigl(\partial_\tau^2 g\bigr)^\top g+(\partial_\tau g)^\top \partial_\tau g,
    \qquad
    D_z S=\bigl(D_z\partial_\tau g\bigr)^\top g+(\partial_\tau g)^\top D_z g .
\]
Evaluating at \((\alpha_0,z_0)\) and using \(g(\alpha_0,z_0)=0\) to kill the leading terms,
\[
        \partial_\tau S(\alpha_0,z_0)=G_\tau^\top G_\tau>0,
        \qquad
        D_zS(\alpha_0,z_0)=G_\tau^\top G_z .
\]
The strict positivity \(\partial_\tau S(\alpha_0,z_0)>0\) means that \(S(\cdot,z_0)\) crosses zero with a positive derivative at \(\alpha_0\); hence on a sufficiently small open neighborhood \(\mathcal V\) of \(\alpha_0\), \(\alpha_0\) is the unique zero of \(S(\cdot,z_0)\) in \(\mathcal V\). The implicit function theorem then gives an open neighborhood of \(z_0\) and a \(C^1\) map \(a(z)\) with \(a(z_0)=\alpha_0\), \(a(z)\in\mathcal V\), and \(S\{a(z),z\}=0\), with \(a(z)\) the unique solution in \(\mathcal V\). Its derivative is
\[
        Da(z_0)=-(G_\tau^\top G_\tau)^{-1}G_\tau^\top G_z.
\]
Because \(\hat\alpha\to_p\alpha_0\), eventually \(\hat\alpha\in\mathcal V\); because the first-order condition \(S(\hat\alpha,\hat z)=0\) holds with probability tending to one and \(\hat z\) lies in the neighborhood of \(z_0\) on which \(a(\cdot)\) is defined, local uniqueness in \(\mathcal V\) forces \(\hat\alpha=a(\hat z)\) with probability tending to one. A first-order expansion of \(a(\hat z)\) around \(z_0\) gives \eqref{eq:app-md-expansion}.
\end{proof}

\subsection{Consistency lemmas}\label{app:consistency-lemmas}

\begin{lemma}[Uniform plug-in continuity]\label{lem:app-uniform-plugin}
Let \(\Theta\) be compact and let \(K\) be a compact convex subset of the finite-dimensional moment space. Suppose \(f(\tau,z)\) is continuously differentiable in \(z\) on an open set containing \(\Theta\times K\). Then there exists \(L<\infty\) such that
\[
    \sup_{\tau\in\Theta}\|f(\tau,z_1)-f(\tau,z_2)\|
    \le L\|z_1-z_2\|,
    \qquad z_1,z_2\in K .
\]
Consequently, if \(\bar Z\xrightarrow{p}z_0\in K\) and \(\Pr(\bar Z\in K)\to1\), then
\[
    \sup_{\tau\in\Theta}\|f(\tau,\bar Z)-f(\tau,z_0)\|\xrightarrow{p}0 .
\]
\end{lemma}

\begin{proof}
Continuity of \(D_zf\) on the compact set \(\Theta\times K\) implies \(\sup_{\Theta\times K}\|D_zf\|<\infty\). The mean-value theorem along the segment joining \(z_1\) and \(z_2\) gives the first display. The stochastic conclusion follows by taking \(z_1=\bar Z\), \(z_2=z_0\), and using \(\Pr(\bar Z\in K)\to1\).
\end{proof}

\begin{lemma}[Uniform spectral separation]\label{lem:app-uniform-spectral-gap}
Under Assumptions~\ref{ass:finite-lda} and \ref{ass:alpha-probes},
\[
    \inf_{\tau\in\Theta}\lambda_k(B_\tau)>0,
    \qquad
    \lambda_{k+1}(B_\tau)=0\quad (k<d),
\]
where the eigenvalues are ordered nonincreasingly. Moreover, there is a neighborhood \(K\) of \(z_0\) and a constant \(c>0\) such that for every \((\tau,z)\in\Theta\times K\),
\[
    \lambda_k\{B_\tau(z)\}\ge c,
    \qquad
    \lambda_{k+1}\{B_\tau(z)\}\le c/2\quad (k<d).
\]
\end{lemma}

\begin{proof}
By \eqref{eq:app-SB}, \(B_\tau=OS_B(\tau)O^\top\), with \(S_B(\tau)\succ0\) for all \(\tau>0\). Since \(O\) has full column rank, \(B_\tau\) has exactly \(k\) positive eigenvalues. The map \(\tau\mapsto B_\tau\) is continuous and \(\Theta\) is compact, so the smallest positive eigenvalue is bounded away from zero. The perturbation statement follows from continuity of ordered eigenvalues and compactness of \(\Theta\).
\end{proof}

\begin{lemma}[Rank consistency]\label{lem:app-rank-consistency}
Under Assumptions~\ref{ass:finite-lda} and \ref{ass:rank-threshold}, the estimator \(\hat k\) in \eqref{eq:est-khat} satisfies \(\Pr(\hat k=k)\to1\).
\end{lemma}

\begin{proof}
The population second moment is \(M_2=O\E(hh^\top)O^\top\), where \(\E(hh^\top)\succ0\). Thus \(\rank(M_2)=k\), \(\lambda_k(M_2)>0\), and \(\lambda_{k+1}(M_2)=0\) if \(k<d\). From the primitive CLT, \(\|\hat M_2-M_2\|=O_p(n^{-1/2})\). Weyl's inequality implies
\[
    \max_{1\le j\le d}|\tilde\lambda_j-\lambda_j(M_2)|=O_p(n^{-1/2}).
\]
With probability tending to one, this maximum is smaller than \(\min\{\lambda_k(M_2)/2,a_n\}\). On that event the first \(k\) sample eigenvalues exceed \(a_n\) for all large \(n\), because \(a_n\to0\), and all remaining eigenvalues are below \(a_n\), because \(\sqrt n a_n\to\infty\). Hence \(\hat k=k\) with probability tending to one.
\end{proof}

\begin{lemma}[Concentration estimator: consistency and linearization]\label{lem:app-alpha-md}
Under Assumptions~\ref{ass:finite-lda}, \ref{ass:alpha-probes}, and \ref{ass:rank-threshold}, the estimator \(\hat\alpha_0\) in \eqref{eq:est-alpha0} satisfies \(\hat\alpha_0\xrightarrow{p}\alpha_0\). Furthermore, with \(g\) defined in \eqref{eq:app-g-finite-map},
\[
    G_\tau:=\partial_\tau g(\alpha_0,z_0),
    \qquad
    G_z:=D_zg(\alpha_0,z_0),
\]
we have \(G_\tau^\top G_\tau>0\) and
\begin{equation}\label{eq:app-alpha-md-linear}
    \sqrt n(\hat\alpha_0-\alpha_0)
    =-\big(G_\tau^\top G_\tau\big)^{-1}G_\tau^\top G_z\sqrt n(\bar Z-z_0)+o_p(1).
\end{equation}
\end{lemma}

\begin{proof}
By Lemma~\ref{lem:app-rank-consistency}, it is enough to work on \(\{\hat k=k\}\).
Lemmas~\ref{lem:app-uniform-spectral-gap} and \ref{lem:app-smooth-estimator-maps} imply that
the commutator map and its first-order-condition map are smooth on a neighborhood of
\(\Theta\times\{z_0\}\) reached by \(\bar Z\) with probability tending to one. Lemma~\ref{lem:app-uniform-plugin} gives
\[
    \sup_{\tau\in\Theta}|\|g(\tau,\bar Z)\|^2-\|g(\tau,z_0)\|^2|=o_p(1).
\]
Write \(Q_0(\tau):=\|g(\tau,z_0)\|^2\) and \(\hat Q(\tau):=\|g(\tau,\bar Z)\|^2\). By Lemma~\ref{lem:finite-probe-local}, \(Q_0\) has unique minimizer \(\alpha_0\) on \(\Theta\) with \(Q_0(\alpha_0)=0\); since \(Q_0\) is continuous and \(\Theta\) is compact, this minimizer is well separated, in the sense that \(\inf\{Q_0(\tau):\tau\in\Theta,\,|\tau-\alpha_0|\ge\epsilon\}>0\) for every \(\epsilon>0\). Using \(\hat Q(\hat\alpha_0)\le\hat Q(\alpha_0)\),
\[
    Q_0(\hat\alpha_0)-Q_0(\alpha_0)
    =\bigl[Q_0(\hat\alpha_0)-\hat Q(\hat\alpha_0)\bigr]
    +\bigl[\hat Q(\hat\alpha_0)-\hat Q(\alpha_0)\bigr]
    +\bigl[\hat Q(\alpha_0)-Q_0(\alpha_0)\bigr]
    \le 2\sup_{\tau\in\Theta}|\hat Q(\tau)-Q_0(\tau)|=o_p(1),
\]
so \(Q_0(\hat\alpha_0)\to_p 0\). The well-separated minimizer property then forces \(\hat\alpha_0\to_p\alpha_0\). Since
\(\alpha_0\in\operatorname{int}(\Theta)\), the sample minimizer is interior with probability tending
to one and hence satisfies the first-order condition
\(S(\hat\alpha_0,\bar Z)=0\). Lemma~\ref{lem:finite-probe-local} also gives \(G_\tau\ne0\).
Applying Lemma~\ref{lem:app-md-expansion} proves \eqref{eq:app-alpha-md-linear}.
\end{proof}

\subsection{Proof of Theorem~\ref{thm:est-main-asymptotics}}\label{app:proof-main-asymptotics}

All statements are conditional on the realized probes and ordering direction. Let \(z_0=Z_0\) denote the population moment vector in \eqref{eq:app-Zi-vector}.

\paragraph{Step 1: rank and concentration.}
Lemma~\ref{lem:app-rank-consistency} gives \(\Pr(\hat k=k)\to1\). On this event, Lemma~\ref{lem:app-alpha-md} gives
\begin{equation}\label{eq:app-alpha-IF}
    \sqrt n(\hat\alpha_0-\alpha_0)
    =D_\alpha\sqrt n(\bar Z-z_0)+o_p(1),
    \qquad
    D_\alpha:=-(G_\tau^\top G_\tau)^{-1}G_\tau^\top G_z .
\end{equation}
Equivalently,
\begin{equation}\label{eq:app-alpha-phi}
    \sqrt n(\hat\alpha_0-\alpha_0)
    =\frac1{\sqrt n}\sum_{i=1}^n\phi_{\alpha,i}+o_p(1),
    \qquad
    \phi_{\alpha,i}:=D_\alpha(Z_i-z_0).
\end{equation}

\paragraph{Step 2: topic estimator.}
Let \(\eta_r=P(\mu)r\). At the truth,
\begin{equation}\label{eq:app-Ho-spectral}
    H^o:=H_{\alpha_0}(\eta_r;z_0)
    =O_r\Lambda_rO_r^+,
    \qquad
    \Lambda_r=\frac{2}{\alpha_0+2}\diag(O_r^\top\eta_r),
\end{equation}
where \(O_r\) is ordered by the entries of \(O^\top\eta_r\). Assumption~\ref{ass:ordering} gives simple nonzero eigenvalues separated from the zero cluster, and the local labeling is by decreasing real eigenvalue within that cluster. Because the population eigenvalues are simple and mutually separated, on a small enough neighborhood of \((z_0,\alpha_0)\) the sample eigenvalues retain this ordering with probability tending to one; the sort permutation is therefore locally constant and contributes no first-order term, and the eigenvector map can be linearized through the implicit function theorem for each ordered eigenpair separately. By Lemma~\ref{lem:app-smooth-estimator-maps}, the local eigenvector map \((z,\tau)\mapsto O_r(z,\tau)\) is continuously differentiable. Therefore
\begin{equation}\label{eq:app-O-expansion-abstract}
    \sqrt n\,\vecop\{\hat O(r)-O_r\}
    =D_O
      \begin{pmatrix}
        \sqrt n(\bar Z-z_0)\\
        \sqrt n(\hat\alpha_0-\alpha_0)
      \end{pmatrix}
      +o_p(1)
\end{equation}
for the derivative \(D_O\) of the normalized eigenvector map at \((z_0,\alpha_0)\). Combining with \eqref{eq:app-alpha-IF}, there is a matrix \(\tilde D_O\) such that
\begin{equation}\label{eq:app-O-expansion}
    \sqrt n\,\vecop\{\hat O(r)-O_r\}
    =\tilde D_O\sqrt n(\bar Z-z_0)+o_p(1)
    =\frac1{\sqrt n}\sum_{i=1}^n\phi_{O,i}(r)+o_p(1),
\end{equation}
where \(\phi_{O,i}(r)=\tilde D_O(Z_i-z_0)\).

\paragraph{Step 3: supervised operator and downstream coefficient.}
Let
\[
    H^y=H_{\alpha_0}^y(z_0)=O_r\Lambda_yO_r^+,
    \qquad
    \Lambda_y=\frac{2}{\alpha_0+2}\diag(\beta_r).
\]
The map \((z,\tau)\mapsto H_\tau^y(z)\) is continuously differentiable by Lemma~\ref{lem:app-smooth-estimator-maps}. Let \(D_H\) denote its derivative at \((z_0,\alpha_0)\), after substituting the influence of \(\hat\alpha_0\) from \eqref{eq:app-alpha-IF}. Then
\begin{equation}\label{eq:app-Hy-expansion}
    \sqrt n(\hat H^y-H^y)
    =D_H\sqrt n(\bar Z-z_0)+o_p(1)
    =\frac1{\sqrt n}\sum_{i=1}^n\Delta_i^y+o_p(1),
\end{equation}
where \(\Delta_i^y=D_H(Z_i-z_0)\). Applying Lemma~\ref{lem:app-beta-cancellation} to the coefficient map gives
\begin{equation}\label{eq:app-beta-expansion}
    \sqrt n\{\hat\beta(r)-\beta_r\}
    =\frac1{\sqrt n}\sum_{i=1}^n\phi_{\beta,i}(r)+o_p(1),
\end{equation}
with
\begin{equation}\label{eq:app-beta-phi}
    \phi_{\beta,i}(r)
    =\frac{\alpha_0+2}{2}\diag\{O_r^+\Delta_i^yO_r\}
      +\frac{\phi_{\alpha,i}}{\alpha_0+2}\beta_r .
\end{equation}
The two terms in \eqref{eq:app-beta-phi} are not double-counting the contribution of \(\hat\alpha_0\). Writing the population map as
\[
    G(\tau,z)=\tfrac{\tau+2}{2}\diag\{O(\tau,z)^+H^y(\tau,z)O(\tau,z)\}
\]
and applying the chain rule \(dG/dz=\partial_z G+\partial_\tau G\cdot D_\alpha\), the diagonal commutator cancellation in Lemma~\ref{lem:app-beta-cancellation} removes the \(\partial_z O\) and \(\partial_\tau O\) contributions inside the diagonal; what remains is
\[
    \frac{dG}{dz}\bigg|_{(\alpha_0,z_0)}
    =\tfrac{\alpha_0+2}{2}\diag\{O_r^+\,D_H\,O_r\}+\tfrac{\beta_r}{\alpha_0+2}D_\alpha,
\]
where \(D_H=\partial_z H^y+(\partial_\tau H^y)D_\alpha\) is the operator derivative used in \eqref{eq:app-Hy-expansion}, and the second piece comes from differentiating the multiplicative factor \((\tau+2)/2\) and using \(\diag\{O_r^+H^yO_r\}=2\beta_r/(\alpha_0+2)\) at the truth. The first term in \eqref{eq:app-beta-phi} carries the influence of \(\hat\alpha_0\) through the operator \(D_H\); the second carries it through the prefactor. No separate term involving \(\phi_{O,i}\) appears because of the diagonal commutator cancellation in Lemma~\ref{lem:app-beta-cancellation}; estimating \(O\) is still needed to form the sample projection and to label the coefficients.

\paragraph{Step 4: joint expansion and variance estimation.}
Combining \eqref{eq:app-alpha-phi}, \eqref{eq:app-O-expansion}, and \eqref{eq:app-beta-expansion}, define
\begin{equation}\label{eq:app-joint-phi}
    \phi_i(r)=
    \begin{pmatrix}
        \phi_{\alpha,i}\\
        \phi_{O,i}(r)\\
        \phi_{\beta,i}(r)
    \end{pmatrix}.
\end{equation}
Then
\begin{equation}\label{eq:app-theta-expansion}
    \sqrt n
    \begin{pmatrix}
        \hat\alpha_0-\alpha_0\\
        \vecop\{\hat O(r)-O_r\}\\
        \hat\beta(r)-\beta_r
    \end{pmatrix}
    =\frac1{\sqrt n}\sum_{i=1}^n\phi_i(r)+o_p(1).
\end{equation}
Conditional on the probes and ordering direction, \(\phi_i(r)\) is a fixed linear transformation of \(Z_i-z_0\). It is therefore iid, mean zero, and has finite second moment. The multivariate CLT applied to \eqref{eq:app-theta-expansion} gives the joint normal limit in Theorem~\ref{thm:est-main-asymptotics}.

For the coefficient block, write \(\phi_{\beta,i}(r)=D_{\beta,r}(Z_i-z_0)\). Then
\[
    V_\beta(r)=D_{\beta,r}\Omega D_{\beta,r}^\top .
\]
The derivative map is continuous at the truth under the same spectral-gap and probe nondegeneracy conditions. Therefore \(\hat D_{\beta,r}\to_pD_{\beta,r}\), and \(\hat\Omega\to_p\Omega\) by \eqref{eq:app-Omega-hat}. Hence
\[
    \hat V_\beta(r)=\hat D_{\beta,r}\hat\Omega\hat D_{\beta,r}^\top
    \xrightarrow{p}V_\beta(r).
\]
The Wald confidence intervals in \eqref{eq:est-ci} follow from Slutsky's theorem. This completes the proof.

\subsection{Implementation of sandwich standard errors}\label{app:implementation-derivatives}

The simulations use the analytic influence-function implementation of the sandwich estimator.  The calculation is fixed-dimensional once the probes, the ordering direction, and any compression matrix have been fixed.

\begin{enumerate}[leftmargin=*,itemsep=2pt]
\item For each document, compute the document-level moment vector used by the estimator. In the no-control case this is \(Z_i^\star=Z_i\) in \eqref{eq:app-Zi-vector}. With observed controls, use the enlarged vector \(Z_i^\star=Z_i^c\), including the fixed-dimensional control moments in Appendix~\ref{app:observed-controls}. With compressed observed controls, form this enlarged vector from the compressed tokens and include the compressed first and corrected second moments used for scale recovery. With unequal document lengths, use the length-normalized version described in Appendix~\ref{app:variable-lengths}.

\item Form \(\bar Z^\star=n^{-1}\sum_i Z_i^\star\) and
\[
    \hat\Omega^\star=\frac1n\sum_{i=1}^n(Z_i^\star-\bar Z^\star)(Z_i^\star-\bar Z^\star)^\top .
\]

\item Compute the point estimates \(\hat\alpha_0\), \(\hat O(r)\), and \(\hat\beta(r)\) from the finite-dimensional maps in Appendix~\ref{app:finite-dimensional-maps}.

\item Let \(g(\tau,z)\) be the stacked finite-probe commutator map in \eqref{eq:app-g-finite-map}.  The derivative of the concentration estimator is the sample analogue of
\[
    D_\alpha
    =-\{G_\tau^\top G_\tau\}^{-1}G_\tau^\top G_z,
    \qquad
    G_\tau=\partial_\tau g(\alpha_0,Z_0),
    \quad
    G_z=\partial_z g(\alpha_0,Z_0).
\]
Equivalently, this is the implicit derivative of the local minimum-distance solution of \(\|g(\tau,z)\|^2\).  The numerical optimizer is used to obtain \(\hat\alpha_0\), but the variance calculation does not differentiate through the optimizer.

\item Differentiate the fixed-\(\alpha_0\) coefficient map \(b_r(z,\alpha_0)\) in \eqref{eq:app-beta-full-map}.  This derivative propagates perturbations through \(B_\tau(z)\), \(A_\tau^y(z)\), the truncated inverse \(B_\tau(z)^{(k)+}\), the normalized right eigenvectors \(O_r(z,\tau)\), and the final diagonal projection.  Let the resulting derivative be \(\hat D_{\beta,r}^{\rm fix}\), and let \(\hat D_{\beta,\alpha}\) be the derivative of \(b_r(\bar Z,\tau)\) with respect to \(\tau\).  The full derivative is
\[
    \hat D_{\beta,r}
    =\hat D_{\beta,r}^{\rm fix}+\hat D_{\beta,\alpha}\hat D_\alpha .
\]

\item If observed controls are used, construct the control-adjusted map in Appendix~\ref{app:observed-controls}. If compressed observed controls are used, first recover the compressed topic scale from \(\hat\mu_R\) and \(\hat B_{\hat\alpha_0,R}\) as in Proposition~\ref{prop:app-compressed-controls}, then construct \(\widehat M_{qh}\), \(\widehat M_{qq}\), \(\hat s_{qY}\), and the control-specific maps \(\hat\psi_{q_a}\). The derivative, denoted \(\hat D^c_{\beta,r}\) in the control-adjusted case, is taken with respect to the enlarged vector \(Z_i^\star\) and includes the scale-recovery and control-adjustment steps.

\item Compute the empirical influence values
\[
    \hat\phi^\star_{\beta,i}(r)=\hat D^\star_{\beta,r}(Z_i^\star-\bar Z^\star),
\]
where \(\hat D^\star_{\beta,r}=\hat D_{\beta,r}\) in the no-control case and \(\hat D^\star_{\beta,r}=\hat D^c_{\beta,r}\) in the control-adjusted case. Report
\[
    \hat V_\beta(r)=\frac1n\sum_{i=1}^n
    \hat\phi^\star_{\beta,i}(r)\hat\phi^\star_{\beta,i}(r)^\top,
    \qquad
    \widehat{\operatorname{se}}\{\hat\beta_j(r)\}
    =\{\hat V_{\beta,jj}(r)/n\}^{1/2}.
\]
\end{enumerate}

The same calculation gives the reported standard error for \(\hat\alpha_0\), using \(\hat\phi_{\alpha,i}=\hat D_\alpha(Z_i^\star-\bar Z^\star)\) with the derivative taken with respect to the word-side components that enter the commutator criterion.

\paragraph{Implementation conventions.}
The simulations and application use raw probe directions drawn from independent standard normal
coordinates and then projected to the empirical mean-orthogonal space by \(P(\hat\mu)\). In the
application, the first-split rank diagnostic uses 20 independent probe sets and five commutator
partner directions for each candidate rank. The reported second-split application table uses probe
seed 13579 and the same number of partner directions. The concentration search interval is
\(\Theta=[0.05,30]\); scalar minimization is initialized by a grid search over this interval. The
ordering direction is generated by the same projected-probe construction. The ordering step selects
the nonzero eigenvalue cluster of the ordering operator; within that selected cluster, eigenpairs are
labeled by decreasing real eigenvalue. Negative selected ordering eigenvalues are not discarded.
In the Monte Carlo Hellinger summaries, estimated topic columns are used only for simulation scoring: columns are
sign-normalized, negative entries are clipped to zero, and the columns are renormalized before
computing Hellinger distances; topic labels are then aligned to the truth by Hungarian matching.

\subsection{Observed controls}\label{app:observed-controls}

This subsection records the extension to fixed-dimensional observed controls. Let
\(q_i\in\R^p\), with fixed \(p\), be observed together with the document and the
response. The word model is unchanged: conditional on \((h_i,q_i)\), the token
variables are generated as in Section~\ref{sec:model}. The response residual
is assumed to be orthogonal to the corresponding word moments conditional on
the observed controls, as implied by \eqref{eq:app-controls-model}. In particular, for
\(j\ne \ell\),
\[
    \E(x_{ij}\mid h_i,q_i)=Oh_i,
    \qquad
    \E(x_{ij}x_{i\ell}^\top\mid h_i,q_i)=(Oh_i)(Oh_i)^\top .
\]
The response restriction is replaced by
\begin{equation}\label{eq:app-controls-model}
    \E(Y_i\mid h_i,q_i,x_{i1},\ldots,x_{iN})
    =\beta^\top h_i+\delta^\top q_i,
\end{equation}
where \(\delta\in\R^p\). The controls may be correlated with the latent topic
shares; independence of \(q_i\) and \(h_i\) is not required. A constant control
is not separately identified from a common shift in the entries of \(\beta\),
because \(\ones^\top h_i=1\). Thus either no intercept is included in \(q_i\),
as assumed below, or an additional normalization on \(\beta\) is imposed.

For any scalar document-level variable \(R_i\), define
\[
    m_R=\E R_i,
    \qquad
    v_R=\E(R_ix_{i1}),
    \qquad
    T^R=\E(R_ix_{i1}x_{i2}^\top),
\]
where the two token positions are distinct. Let
\begin{equation}\label{eq:app-controls-AR}
\begin{aligned}
    A_\tau^R
    :={}&T^R-\frac{\tau}{\tau+2}
       \{v_R\mu^\top+\mu v_R^\top+m_RM_2\}  \\
    &\quad
     +\frac{2\tau^2}{(\tau+1)(\tau+2)}m_R\mu\mu^\top,
    \qquad
    H_\tau^R:=A_\tau^RB_\tau^+ .
\end{aligned}
\end{equation}
At the true concentration define the linear functional
\begin{equation}\label{eq:app-controls-psi}
    \psi(R)
    :=\frac{\alpha_0+2}{2}
      \diag\{O^+H_{\alpha_0}^RO\}\in\R^k .
\end{equation}
If \(R_i=a^\top h_i\), Theorem~\ref{thm:id-beta-supervised} gives
\(\psi(R)=a\). For a general control coordinate \(q_a\), the matrix
\(O^+H_{\alpha_0}^{q_a}O\) need not be diagonal; only its diagonal is used.
Writing
\[
    \psi_Y:=\psi(Y),
    \qquad
    \Psi_q:=\{\psi(q_1),\ldots,\psi(q_p)\}\in\R^{k\times p},
\]
linearity of \(A_\tau^R\), \(H_\tau^R\), and \(\psi(R)\), together with
\eqref{eq:app-controls-model}, gives
\begin{equation}\label{eq:app-controls-first-eq}
    \psi_Y=\beta+\Psi_q\delta .
\end{equation}
The error component contributes zero because its conditional mean given
\((h_i,q_i,x_{i1},\ldots,x_{iN})\) is zero, and hence its contributions to
\(m_R\), \(v_R\), and \(T^R\) vanish.

A second set of observed moments separates \(\beta\) from \(\delta\). Let
\[
    M_{qh}:=\E(q_ih_i^\top),
    \qquad
    M_{qq}:=\E(q_iq_i^\top),
    \qquad
    s_{qY}:=\E(q_iY_i).
\]
The cross-moment \(M_{qh}\) is identified from words and controls, since
\begin{equation}\label{eq:app-controls-Mqh}
    \E(q_ix_{i1}^\top)=M_{qh}O^\top,
    \qquad
    M_{qh}=\E(q_ix_{i1}^\top)O^{+\top} .
\end{equation}
Multiplying \eqref{eq:app-controls-model} by \(q_i\) and taking expectations
gives
\begin{equation}\label{eq:app-controls-second-eq}
    s_{qY}=M_{qh}\beta+M_{qq}\delta .
\end{equation}
Combining \eqref{eq:app-controls-first-eq} and
\eqref{eq:app-controls-second-eq}, define
\begin{equation}\label{eq:app-controls-Gamma}
    \Gamma_q:=M_{qq}-M_{qh}\Psi_q .
\end{equation}
If \(\Gamma_q\) is nonsingular, then
\begin{align}
    \delta
    &=\Gamma_q^{-1}\{s_{qY}-M_{qh}\psi_Y\},
      \label{eq:app-controls-delta}\\
    \beta
    &=\psi_Y-\Psi_q\delta .
      \label{eq:app-controls-beta}
\end{align}
Equivalently,
\begin{equation}\label{eq:app-controls-beta-closed}
    \beta
    =\{I_k+\Psi_q\Gamma_q^{-1}M_{qh}\}\psi_Y
     -\Psi_q\Gamma_q^{-1}s_{qY} .
\end{equation}
This identifies the topic coefficient in the presence of observed controls. The
nonsingularity condition excludes exact collinearity between the observed
controls and the latent-topic component after the moment transformation above.

The sample analogue is obtained by replacing each population moment by its
empirical counterpart. For a scalar document variable \(R_i\), put
\[
    \hat m_R=n^{-1}\sum_{i=1}^nR_i,
    \qquad
    \hat v_R=n^{-1}\sum_{i=1}^nR_i\hat\mu_i,
    \qquad
    \hat T^R=n^{-1}\sum_{i=1}^nR_i\hat M_{2,i}.
\]
Construct \(\hat A_{\hat\alpha_0}^R\) from
\eqref{eq:app-controls-AR}, with population moments replaced by sample moments,
and define
\begin{equation}\label{eq:app-controls-psihat}
    \hat\psi_R(r)
    :=\frac{\hat\alpha_0+2}{2}
       \diag\{\hat O(r)^+\hat A_{\hat\alpha_0}^R
       \hat B_{\hat\alpha_0,\hat k}^+\hat O(r)\} .
\end{equation}
Let
\[
    \hat\psi_Y(r):=\hat\psi_R(r)\big|_{R=Y},
    \qquad
    \hat\Psi_q(r):=\{\hat\psi_{q_1}(r),\ldots,\hat\psi_{q_p}(r)\},
\]
\[
    \hat M_{qh}
    :=\left(n^{-1}\sum_{i=1}^nq_i\hat\mu_i^\top\right)
      \hat O(r)^{+\top},
    \qquad
    \hat M_{qq}:=n^{-1}\sum_{i=1}^nq_iq_i^\top,
    \qquad
    \hat s_{qY}:=n^{-1}\sum_{i=1}^nq_iY_i .
\]
With \(\hat\Gamma_q:=\hat M_{qq}-\hat M_{qh}\hat\Psi_q(r)\), define
\begin{align}
    \hat\delta
    &=\hat\Gamma_q^{-1}\{\hat s_{qY}-\hat M_{qh}\hat\psi_Y(r)\},
      \label{eq:app-controls-deltahat}\\
    \hat\beta_c(r)
    &=\hat\psi_Y(r)-\hat\Psi_q(r)\hat\delta .
      \label{eq:app-controls-betahat}
\end{align}
Here \(\hat\beta_c(r)\) is the control-adjusted topic coefficient in the same
ordering as \(\hat O(r)\).

The asymptotic theory is an immediate delta-method extension of
Theorem~\ref{thm:est-main-asymptotics}. Enlarge the document-level vector
\(Z_i\) in \eqref{eq:app-Zi-vector} to include the fixed-dimensional moments
needed above, for example
\[
    q_i,
    \quad q_iq_i^\top,
    \quad q_iY_i,
    \quad q_i\hat\mu_i^\top,
    \quad q_{ia}\hat\mu_i,
    \quad q_{ia}\hat M_{2,i}\quad (a=1,\ldots,p).
\]
Call the enlarged vector \(Z_i^c\), with mean \(Z_0^c\) and covariance
\(\Omega_c\). If \(\E\|Z_i^c\|^2<\infty\) and \(\Gamma_q\) is nonsingular,
then the maps in \eqref{eq:app-controls-psihat}--\eqref{eq:app-controls-betahat}
are continuously differentiable at the truth under the same rank, probe, and
ordering conditions as before. Therefore, conditionally on the realized probes
and ordering direction,
\begin{equation}\label{eq:app-controls-clt}
    \sqrt n
    \begin{pmatrix}
        \hat\beta_c(r)-\beta_r\\
        \hat\delta-\delta
    \end{pmatrix}
    =D_c\sqrt n(\bar Z^c-Z_0^c)+o_p(1)
    \rightsquigarrow N(0,D_c\Omega_cD_c^\top),
\end{equation}
where \(D_c\) is the derivative of the displayed plug-in map. A consistent
sandwich estimator is obtained by evaluating \(D_c\) at the sample moments and
replacing \(\Omega_c\) by the empirical covariance of \(Z_i^c\).

\subsection{Linear compression and split-sample PCA preprocessing}\label{app:pca-compression}

This subsection records the finite-dimensional compression properties used when the vocabulary dimension is
large. The results are not growing-\(d\) theorems. Rather, they say that the moment construction may be
applied after a fixed-dimensional linear compression, provided the compression does not lose any topic
direction. In algebraic terms, for \(R\in\R^{d\times m}\) with fixed \(m\ge k\), the compressed loading
matrix \(R^\top O\) must have full column rank. This allows the second- and third-order moment
calculations to be carried out in dimension \(m\).

Let \(R\in\R^{d\times m}\), with fixed \(m\ge k\), be a deterministic matrix and define compressed
token vectors
\[
    z_{ij}=R^\top x_{ij}\in\R^m,
    \qquad
    \widetilde O_R:=R^\top O\in\R^{m\times k}.
\]
The columns of \(\widetilde O_R\) are not probability vectors. In sample implementations, compressed
eigenvectors are therefore first recovered only up to nonzero column rescaling. The no-control
coefficient formula below is invariant to such rescalings, but observed-control adjustments also use
the topic-share scale through \(M_{qh}=\E(q_i h_i^\top)\). Proposition~\ref{prop:app-compressed-controls}
therefore gives a moment-based scale-recovery step for compressed observed-control inference.

For the compressed variables write
\[
    \widetilde\mu_R=\E z_{i1},\qquad
    \widetilde M_{2,R}=\E(z_{i1}z_{i2}^\top),\qquad
    \widetilde T_R(a)=\E\{z_{i1}z_{i2}^\top\langle z_{i3},a\rangle\},
\]
where \(a\in\R^m\). Let \(\widetilde B_{\tau,R}\), \(\widetilde A_{\tau,R}(a)\),
\(\widetilde H_{\tau,R}(a)\), and \(\widetilde H^y_{\tau,R}\) denote the analogues of
\eqref{eq:id-Btau}, \eqref{eq:id-Atau}, \eqref{eq:id-Htau}, and \eqref{eq:id-Hytau} formed from
the compressed moments. When \(\widetilde O_R\) has full column rank, the corrected second moment
has rank \(k\), and the rank-\(k\) Moore--Penrose inverse is used in the compressed operator.

\begin{proposition}[Invariance to admissible linear compression]\label{prop:app-linear-compression}
Assume the model of Section~\ref{sec:model}, with \(O\) full column rank. Let
\(R\in\R^{d\times m}\), \(m\ge k\), satisfy \(\rank(R^\top O)=k\). Then the compressed tokens
\(z_{ij}=R^\top x_{ij}\) satisfy the same cross-token moment identities with loading matrix
\(\widetilde O_R=R^\top O\). In particular, at the true concentration,
\begin{equation}\label{eq:app-compressed-H}
    \widetilde H_{\alpha_0,R}(a)
    = \widetilde O_R
      \left\{\frac{2}{\alpha_0+2}\diag(\widetilde O_R^\top a)\right\}
      \widetilde O_R^+,
\end{equation}
and
\begin{equation}\label{eq:app-compressed-Hy}
    \widetilde H^y_{\alpha_0,R}
    = \widetilde O_R
      \left\{\frac{2}{\alpha_0+2}\diag(\beta)\right\}
      \widetilde O_R^+.
\end{equation}
Consequently, the coefficient is identified from compressed moments by
\begin{equation}\label{eq:app-compressed-beta}
    \beta
    =\frac{\alpha_0+2}{2}
      \diag\{\widetilde O_R^+\widetilde H^y_{\alpha_0,R}\widetilde O_R\}.
\end{equation}
The commutativity characterization of \(\alpha_0\) also holds in the compressed coordinates whenever
\(k\ge3\) and the compressed probe directions satisfy the same non-collinearity condition as in
Theorem~\ref{thm:id-alpha-commutativity}.
\end{proposition}

\begin{proof}
Conditional on \(h_i\),
\[
    \E(z_{ij}\mid h_i)=R^\top Oh_i=\widetilde O_Rh_i.
\]
For distinct token positions, conditional independence gives
\[
    \E(z_{ij}z_{i\ell}^\top\mid h_i)=\widetilde O_Rh_ih_i^\top\widetilde O_R^\top,
\]
and the analogous third cross-token identity. Equivalently,
\[
    \widetilde\mu_R=R^\top\mu,
    \qquad
    \widetilde M_{2,R}=R^\top M_2R,
    \qquad
    \widetilde T_R(a)=R^\top T(Ra)R,
\]
and therefore
\[
    \widetilde A_{\tau,R}(a)=R^\top A_\tau(Ra)R,
    \qquad
    \widetilde B_{\tau,R}=R^\top B_\tau R.
\]
At \(\tau=\alpha_0\), Lemma~\ref{lem:id-corrected-factorization} gives
\[
    \widetilde B_{\alpha_0,R}
    =\widetilde O_R\frac{D}{C_2}\widetilde O_R^\top,
    \qquad
    \widetilde A_{\alpha_0,R}(a)
    =\widetilde O_R\frac{2}{C_3}
      \diag\{\alpha\circ(\widetilde O_R^\top a)\}\widetilde O_R^\top .
\]
Since \(\widetilde O_R\) has full column rank and \(D\) is positive definite,
\(\widetilde B_{\alpha_0,R}\) has rank \(k\). Multiplying
\(\widetilde A_{\alpha_0,R}(a)\) by the rank-\(k\) Moore--Penrose inverse of
\(\widetilde B_{\alpha_0,R}\) gives \eqref{eq:app-compressed-H}. The supervised identity
\eqref{eq:app-compressed-Hy} follows in the same way from the supervised bridge identities implied
by \eqref{eq:response-token-orthogonality}, exactly as in the proof of Theorem~\ref{thm:id-beta-supervised}.
Equation~\eqref{eq:app-compressed-beta} follows by premultiplying by \(\widetilde O_R^+\),
postmultiplying by \(\widetilde O_R\), and taking the diagonal. The commutativity argument is
identical to the proof of Theorem~\ref{thm:id-alpha-commutativity}, with \(O\) replaced by the
full-column-rank matrix \(\widetilde O_R\). Since the map \(a\mapsto \widetilde O_R^\top a\) is
onto \(\R^k\), the required latent non-collinearity conditions are unchanged.
\end{proof}

\begin{proposition}[Observed controls after compressed scale recovery]\label{prop:app-compressed-controls}
Assume the observed-control model of Appendix~\ref{app:observed-controls}. Let
\(R\in\R^{d\times m}\), with fixed \(m\ge k\), satisfy \(\rank(R^\top O)=k\), and write
\(\widetilde O_R=R^\top O\). Suppose the compressed spectral step has recovered the compressed topic
directions in the correct order but with arbitrary nonzero column scale,
\[
        \bar O_R=\widetilde O_R C,
        \qquad C=\diag(c_1,\ldots,c_k),
\]
where \(C\) is nonsingular. Let
\[
        \widetilde\mu_R=\E z_{i1},
        \qquad
        \widetilde B_{\alpha_0,R}
        =\E(z_{i1}z_{i2}^\top)-\frac{\alpha_0}{\alpha_0+1}\widetilde\mu_R\widetilde\mu_R^\top,
\]
and set \(L=\bar O_R^+\). For \(j=1,\ldots,k\), define
\[
        \theta_j=(L\widetilde\mu_R)_j,
        \qquad
        b_j=\{L\widetilde B_{\alpha_0,R}L^\top\}_{jj},
        \qquad
        c_j^0=\frac{\theta_j}{(\alpha_0+1)b_j}.
\]
Then, at the population value, \(c_j^0=c_j\) for every \(j\). Hence
\[
        \bar O_R\diag(c_1^0,\ldots,c_k^0)^{-1}=\widetilde O_R,
\]
so the correctly scaled compressed loading matrix is identified from compressed first and second
corrected moments.

Let \(U_i\) be any scalar document-level variable and define
\[
        \widetilde\psi(U)
        =\frac{\alpha_0+2}{2}\diag\{\widetilde O_R^+\widetilde H^U_{\alpha_0,R}\widetilde O_R\},
\]
where \(\widetilde H^U_{\alpha_0,R}\) is the compressed corrected operator formed with \(U_i\) in
place of \(Y_i\). Put
\(\widetilde\psi_Y=\widetilde\psi(Y)\) and
\(\widetilde\Psi_q=\{\widetilde\psi(q_1),\ldots,\widetilde\psi(q_p)\}\). Then
\[
        \widetilde\psi_Y=\beta+\widetilde\Psi_q\delta .
\]
Moreover, with \(\widetilde M_{qz}=\E(q_i z_{i1}^\top)\),
\[
        M_{qh}=\widetilde M_{qz}\,\widetilde O_R^{+\top}.
\]
Therefore, if \(\Gamma_q=M_{qq}-M_{qh}\widetilde\Psi_q\) is nonsingular, then
\[
        \delta=\Gamma_q^{-1}\{s_{qY}-M_{qh}\widetilde\psi_Y\},
        \qquad
        \beta=\widetilde\psi_Y-\widetilde\Psi_q\delta .
\]
\end{proposition}

\begin{proof}
The compressed word model is
\[
        \E(z_{ij}\mid h_i)=R^\top Oh_i=\widetilde O_Rh_i .
\]
Thus the compressed first moment is \(\widetilde\mu_R=\widetilde O_R\pi\), where
\(\pi=\alpha/\alpha_0\), and the corrected compressed second moment at the true concentration is
\[
        \widetilde B_{\alpha_0,R}
        =\widetilde O_R \frac{D}{\alpha_0(\alpha_0+1)}\widetilde O_R^\top .
\]
Since \(\bar O_R=\widetilde O_RC\) and \(\widetilde O_R\) has full column rank,
\[
        \bar O_R^+\widetilde\mu_R=C^{-1}\pi,
\]
so \(\theta_j=\pi_j/c_j\). Similarly,
\[
        \bar O_R^+\widetilde B_{\alpha_0,R}\bar O_R^{+\top}
        =C^{-1}\frac{D}{\alpha_0(\alpha_0+1)}C^{-1},
\]
and hence
\[
        b_j=\frac{\alpha_j}{\alpha_0(\alpha_0+1)c_j^2}
        =\frac{\pi_j}{(\alpha_0+1)c_j^2}.
\]
Therefore
\[
        \frac{\theta_j}{(\alpha_0+1)b_j}
        =\frac{\pi_j/c_j}{\pi_j/c_j^2}=c_j,
\]
which proves the scale-recovery claim.

For any scalar \(U_i\), the compressed corrected moments satisfy the same factorization as the
uncompressed moments, with \(O\) replaced by \(\widetilde O_R\). By linearity of the corrected
operator in \(U\), and by the response model with the low-order residual contribution equal to zero,
\[
        \widetilde\psi_Y=\beta+\widetilde\Psi_q\delta .
\]
The compressed control-word cross moment is
\[
        \widetilde M_{qz}=\E(q_i z_{i1}^\top)=\E(q_i h_i^\top)\widetilde O_R^\top
        =M_{qh}\widetilde O_R^\top .
\]
Multiplying on the right by \(\widetilde O_R^{+\top}\) gives
\(\widetilde M_{qz}\widetilde O_R^{+\top}=M_{qh}\). The two equations
\[
        \widetilde\psi_Y=\beta+\widetilde\Psi_q\delta,
        \qquad
        s_{qY}=M_{qh}\beta+M_{qq}\delta
\]
then yield the displayed formulas for \(\delta\) and \(\beta\) when
\(\Gamma_q=M_{qq}-M_{qh}\widetilde\Psi_q\) is nonsingular.
\end{proof}

The sample version of Proposition~\ref{prop:app-compressed-controls} replaces the population moments by compressed empirical moments and \(\alpha_0\) by \(\hat\alpha_0\). Provided the scale denominators \(b_j\) and the smallest singular value of \(\widetilde O_R\) are bounded away from zero, and \(\Gamma_q\) is nonsingular, the scale-recovery and observed-control maps are smooth finite-dimensional maps of the empirical moments. The same delta-method argument as in Appendix~\ref{app:observed-controls} therefore applies after enlarging the moment vector to include the compressed control moments; the derivative includes the scale-recovery step.

The proposition is deterministic in \(R\). It therefore leads to a simple sample-splitting construction.
Split the documents into two independent parts, \(I_1\) and \(I_2\), with \(|I_s|=n_s\) and
\(n_s/n\) bounded away from zero. Use \(I_1\) to estimate a \(d\times m\) orthonormal matrix
\(\hat R_1\), with fixed \(m\ge k\). For example, \(\hat R_1\) may be the matrix of the leading
\(m\) eigenvectors of the symmetrized first-half uncentered cross-token moment. Then use only
\(I_2\) to compute compressed tokens \(\hat z_{ij}=\hat R_1^\top x_{ij}\) and apply the moment
estimator in dimension \(m\), using a rank-\(k\) truncated inverse of the compressed corrected
second moment.

\begin{proposition}[Split-sample compression]\label{prop:app-split-compression}
Assume the conditions of Theorem~\ref{thm:est-main-asymptotics}, and suppose \(k\) is treated as
known in the compressed analysis. Let \(m\ge k\) be fixed. Let \(\hat R_1\in\R^{d\times m}\) be
estimated from \(I_1\), independently of the second-half moments, and suppose
\begin{equation}\label{eq:app-compression-admissible}
    \Pr\{\sigma_{\min}(\hat R_1^\top O)>c\}\to1
\end{equation}
for some \(c>0\). Conditional on \(\hat R_1\), apply the estimator of Section~\ref{sec:estimation}
to the compressed second-half tokens \(\hat z_{ij}=\hat R_1^\top x_{ij}\), with rank \(k\), with the
same concentration search, with compressed probe and ordering directions satisfying the
non-degeneracy conditions of Theorem~\ref{thm:est-main-asymptotics}, and with compressed
eigenvectors normalized by a fixed smooth rule rather than by the vocabulary-simplex rule. Let
\(\beta_{\hat R_1,r}\) denote \(\beta\) in the ordering induced by the realized compressed ordering
operator. Then, on events whose probability tends to one,
\begin{equation}\label{eq:app-compressed-al}
    \sqrt{n_2}\{\hat\beta_{\hat R_1}(r)-\beta_{\hat R_1,r}\}
    =\frac{1}{\sqrt{n_2}}\sum_{i\in I_2}\phi_{\hat R_1,i}(r)+o_p(1\mid \hat R_1),
\end{equation}
where \(\E\{\phi_{\hat R_1,i}(r)\mid \hat R_1\}=0\) and
\(\E\{\|\phi_{\hat R_1,i}(r)\|^2\mid \hat R_1\}<\infty\). Thus the first-half
estimation error in \(\hat R_1\) does not appear as an additional first-order
term in the second-half expansion. Conditional sandwich standard errors computed
from the compressed second-half moments are valid for the conditional law given
\(\hat R_1\). If the corresponding conditional covariance converges in probability to a nonrandom limit \(V_R\), then \(\sqrt{n_2}\{\hat\beta_{\hat R_1}(r)-\beta_{\hat R_1,r}\}\rightsquigarrow N(0,V_R)\) unconditionally. If fixed-dimensional observed controls are included and the compressed scale-recovery step in Proposition~\ref{prop:app-compressed-controls} is used, the same conditional expansion holds for the control-adjusted estimator after enlarging the second-half moment vector as in Appendix~\ref{app:observed-controls}, provided the population scale-recovery denominators are bounded away from zero and the corresponding \(\Gamma_q(\hat R_1)\) is nonsingular. The derivative used for sandwich standard errors includes the scale-recovery map.
\end{proposition}

\begin{proof}
On the event in \eqref{eq:app-compression-admissible}, the matrix
\(\hat R_1^\top O\) has full column rank. Conditional on the first half, \(\hat R_1\)
is fixed, and the second-half compressed observations are iid with the same
cross-token moment identities and loading matrix \(\hat R_1^\top O\). Proposition~\ref{prop:app-linear-compression}
shows that the population target of the compressed procedure is the same
\(\beta\) as in the original coordinates, with ordering determined by the realized compressed ordering operator. The fixed-dimensional asymptotic
argument of Theorem~\ref{thm:est-main-asymptotics} can therefore be applied
conditionally on \(\hat R_1\), with dimension equal to \(m\) and rank \(k\). This gives
\eqref{eq:app-compressed-al}. Since the equality of targets holds for every
admissible value of \(R\), the random first-half compression affects conditioning
and finite-sample stability but contributes no separate first-order derivative
term to the second-half estimator. The observed-control statement follows by applying the same conditional argument to the smooth compressed scale-recovery and control-adjustment map in Proposition~\ref{prop:app-compressed-controls}.
\end{proof}

The observed-control clause is a conditional delta-method statement. Conditional on \(\hat R_1\), if \(\sigma_{\min}(\hat R_1^\top O)\) is bounded away from zero, the population scale-recovery denominators \(b_j(\hat R_1)\) are bounded away from zero, and \(\Gamma_q(\hat R_1)\) is nonsingular, then the enlarged-moment expansion in \eqref{eq:app-controls-clt} applies on the second split after replacing the uncompressed topic map by the compressed scale-recovery and control-adjustment map of Proposition~\ref{prop:app-compressed-controls}. The empirical sandwich derivative includes this scale-recovery step.

\paragraph{Proof of Corollary~\ref{cor:compressed-control-inference}.}
Conditional on the first split, the realized compression matrix, selected rank,
probes, and ordering direction are fixed. Proposition~\ref{prop:app-linear-compression}
gives the compressed corrected-moment identities when the retained subspace is
admissible, Proposition~\ref{prop:app-compressed-controls} gives the compressed
scale recovery and observed-control map, and Proposition~\ref{prop:app-split-compression}
applies the fixed-dimensional expansion to the independent second-split moments.
Using the enlarged vector \(Z^c_{\hat R_1,i}\) from Appendix~\ref{app:observed-controls}
with compressed tokens gives the derivative \(D^c_{\beta,r}\) and the conditional
influence representation in \eqref{eq:est-compressed-control-corollary}. The
sandwich consistency follows from the conditional iid law of the second-split
moment vectors and continuity of the derivative on the stated nonsingular stratum.

For PCA based on the uncentered cross-token moment, condition
\eqref{eq:app-compression-admissible} is natural. The population matrix
\[
    M_2=O\E(h_ih_i^\top)O^\top
\]
is positive semidefinite with rank \(k\) and column space \(\mathrm{span}(O)\).
When \(d\) is fixed, the symmetrized first-half estimator is root-\(n_1\)
consistent, so its leading \(k\)-dimensional eigenspace is consistent for
\(\mathrm{span}(O)\), using the population eigengap at \(k\). If \(m\ge k\), the
sample leading \(m\)-dimensional subspace contains the sample leading
\(k\)-dimensional eigenspace by construction. Thus
\(\sigma_{\min}(\hat R_1^\top O)\) is bounded away from zero with probability
tending to one. No population eigengap at \(m\) is required. More generally,
any PCA or subspace estimator whose retained subspace contains a consistent
estimate of \(\mathrm{span}(O)\) yields an orthonormal matrix \(\hat R_1\) that obeys
\eqref{eq:app-compression-admissible} with probability tending to one. In
contrast, the centered distinct-token covariance
\(M_2-\mu\mu^\top=O\Var(h_i)O^\top\) has rank at most \(k-1\), because
\(\ones^\top h_i=1\). For a topic-space compression, one should use an uncentered cross-token
second moment, or augment the centered covariance directions with the mean direction.

\begin{remark}[No sample splitting]\label{rem:app-no-split-compression}
The same invariance suggests that using the same observations to estimate
\(R\) and the compressed moments should also have the same first-order limit:
at the population value the coefficient functional is constant over all
admissible \(R\), so its derivative with respect to \(R\) is zero. A formal
no-split proof would nevertheless need to handle the empirical eigenspace map
and stochastic equicontinuity of the projected moment functions. The split-sample
statement above avoids these additional technicalities.
\end{remark}

\subsection{Document lengths other than a common fixed \texorpdfstring{\(N\)}{N}}\label{app:variable-lengths}

The main text assumes a common document length to avoid notational clutter. If observed document lengths \(N_i\ge3\) are allowed, the empirical moments should be normalized document by document:
\[
\begin{aligned}
    \hat\mu_i&=\frac{1}{N_i}\sum_{a=1}^{N_i}x_{ia},\\
    \hat M_{2,i}&=\frac{1}{N_i(N_i-1)}\sum_{a\ne b}x_{ia}x_{ib}^\top,\\
    \hat{\mathcal T}_i
    &=\frac{1}{N_i(N_i-1)(N_i-2)}
      \sum_{a,b,c\;\mathrm{distinct}}x_{ia}\otimes x_{ib}\otimes x_{ic}.
\end{aligned}
\]
Throughout this paper we use the variable-length formulas only under the same sampled-document LDA moment structure as in the main model: the sampled documents have the maintained Dirichlet topic-mixture distribution, and the response orthogonality restrictions hold for the length-normalized moments. Under these conditions the document-level averages above have the same population targets as their common-length analogues. The fixed-dimensional CLT and delta-method arguments are unchanged provided the document-level vector \(Z_i\) has finite second moment.

\section*{Data availability}
The bibliographic metadata and citation counts used in Section~\ref{sec:application}
were obtained from the public OpenAlex Works API, accessed on 25 May 2026. The article texts used for the
application are title and abstract fields reconstructed from OpenAlex metadata.

\bibliographystyle{plainnat}
\bibliography{references}

\end{document}